\documentclass[11pt]{article}
\pdfoutput=1

\input{preamble}

\newcommand{\spc}{\ensuremath{\mathrm{\textsc{spc}}}}
\newcommand{\dist}{\ensuremath{\mathrm{dist}}}
\newcommand{\diam}{\ensuremath{\mathrm{diam}}}

\DeclareMathOperator{\E}{\textbf{E}}

\title{A ${(1+\eps)}$-Embedding of Low Highway Dimension Graphs into 
Bounded Treewidth Graphs\thanks{All authors of this paper were supported by
  NSERC's Discovery Grant Program}}
\author[1]{Andreas Emil Feldmann\thanks{Supported by ERC Starting Grant 
PARAMTIGHT (No.~280152), and by project CE-ITI (GA\v{C}R no.~P202/12/G061) of 
the Czech Science Foundation.}}
\author[2]{Wai Shing Fung}
\author[2]{Jochen K\"{o}nemann\thanks{Supported by the
    Hausdorff Research Institute for Mathematics and the Research
    Institute for Discrete Mathematics in Bonn, Germany}}
\author[2]{Ian Post}
\affil[1]{KAM, Charles University, Prague CZ, \texttt{feldmann.a.e@gmail.com}}
\affil[2]{Department of Combinatorics and Optimization,
University of Waterloo, CA,
\texttt{\{wsfung,jochen\}@uwaterloo.ca, ian@ianpost.org}}
\date{}

\begin{document}
\renewcommand*{\sectionautorefname}{Section}
\renewcommand*{\subsectionautorefname}{Section}
\renewcommand*{\algorithmautorefname}{Algorithm}

\maketitle
\begin{abstract}
  Graphs with bounded \emph{highway dimension} were introduced by
  Abraham et al.~[SODA~2010] as a model of transportation networks. We
  show that any such graph can be embedded into a distribution over
  bounded treewidth graphs with arbitrarily small distortion. More
  concretely, given a weighted graph $G=(V,E)$ of constant highway
  dimension, we show how to randomly compute a weighted graph
  $H=(V,E')$ that distorts shortest
  path distances of $G$ by at most a $1+\eps$
  factor in expectation, and whose treewidth is polylogarithmic in the
  aspect ratio of $G$. Our probabilistic embedding implies
  quasi-polynomial time approximation schemes for a number of
  optimization problems that naturally arise in transportation
  networks, including Travelling Salesman, Steiner Tree,
and Facility Location.

To construct our embedding for low highway dimension graphs we extend Talwar's 
[STOC~2004] embedding of low doubling dimension metrics into bounded treewidth 
graphs, which generalizes known results for Euclidean metrics. We add several 
non-trivial ingredients to Talwar's techniques, and in particular thoroughly 
analyse the structure of low highway dimension graphs. Thus we demonstrate that 
the geometric toolkit used for Euclidean metrics extends beyond the class of 
low doubling metrics.
\end{abstract}

\section{Introduction}

In \cite{bast2007transit,bast2009ultrafast}, Bast et al.\ studied shortest-path 
computations in road networks and observed that such networks are highly 
structured: there is a sparse set of {\em transit} or {\em access} nodes such 
that when travelling from any point $A$ to a distant location $B$ along a 
shortest path, one will visit at least one of these nodes. The authors presented 
a shortest-path algorithm (called {\em transit node routing}) that capitalizes 
on this structure in road networks and demonstrated experimentally that it 
improves over previously best algorithms by several orders of magnitude. 
Motivated by Bast et~al.'s work (among others), 
\citet{abraham2010highway,abraham2011vc,abraham2010highway2} introduced a formal 
model for transportation networks and defined the notion of {\em highway 
dimension}. Informally speaking, an edge-weighted graph $G=(V,E)$ has small {\em 
highway dimension} if, for any {\em scale} $r \geq 0$ and for all vertices $v 
\in V$, shortest paths of length at least $r$ that are close (in terms of $r$) 
to $v$ are {\em hit} by a small set of {\em hub} vertices. In the following 
formal definition, if $\dist(u,v)$ denotes the shortest-path distance between 
vertices $u$ and~$v$, let $B_r(v)= \{u \in V | \dist(u,v) \le r\}$ be the {\em 
ball} of radius $r$ centred at~$v$. We will also say that a path $P$
{\em lies inside $B_r(v)$} if all its vertices lie inside $B_r(v)$. 

\begin{dfn}\label{dfn:hd}
  The \emph{highway dimension} of a graph $G$ is the smallest
  integer~$k$ such that, for some universal constant \mbox{$c\geq 4$},
  for every $r\in \mathbb{R}^+$, and every ball $B_{cr}(v)$ of radius
  $cr$, there are at most $k$ vertices in~$B_{cr}(v)$ hitting all
  shortest paths of length more than $r$ that lie in $B_{cr}(v)$.
\end{dfn}

Rather than working with the above definition directly, we often consider the 
closely related notion of {\em shortest path covers} (also introduced 
in~\cite{abraham2010highway}). 

\begin{dfn}\label{dfn:spc}
  For a graph $G$ and $r\in\mathbb{R}^+$, a \emph{shortest path cover}
  \mbox{$\spc(r)\subseteq V$} is a set of \emph{hubs} that intersect all
  shortest paths of length in $(r,cr/2]$ of $G$. Such a cover is
  called {\em locally $s$-sparse} for scale~$r$, if no ball of radius
  $cr/2$ contains more than~$s$ vertices from $\spc(r)$. 
\end{dfn}

In particular, a graph with highway dimension $k$ can be seen to have
a \emph{locally $k$-sparse} shortest path cover for any scale
$r$~\cite{abraham2010highway} (using the same constant $c$ in \autoref{dfn:hd} 
and \autoref{dfn:spc}). In both definitions above, \citet{abraham2010highway} 
specifically chose $c=4$ but also note that this choice is, to some extent, 
arbitrary. In the present paper, the flexibility of being able to choose a 
slightly larger value of $c$ is crucial as we will explain shortly. In the 
following, we will let $\lambda=c-4$ and call it the {\em violation} of Abraham 
et al.'s original definition. While we believe that a small positive violation 
does not stray from the intended meaning of highway dimension, we also point out 
that there are graphs whose highway dimension is highly sensitive to the value 
of $c$, as we explain in \autoref{sec:alt-defs}. Hence this is not an entirely 
innocuous change.

\citet{abraham2010highway,abraham2011vc,abraham2010highway2} focused
on the shortest-path problem and formally investigated the performance
of various prominent heuristics as a function of the highway dimension
of the input graph. They also pointed out that, ``conceivably,
better algorithms for other [optimization] problems can be developed
and analysed under the small highway dimension assumption''. This
statement is the starting point of this paper.

We study three prominent NP-hard optimization problems that arise naturally in 
transportation networks: {\em Travelling Salesman}, {\em Steiner Tree} and {\em 
Facility Location} (see \autoref{sec:ptas} for formal definitions). Each of 
these was first studied in the context of transportation networks, and as we 
will show they admit quasi-polynomial time approximation schemes (QPTASs) on 
graphs with bounded highway dimension. Our work thereby provides a 
complexity-theoretic separation between the class of low highway dimension 
graphs and general graphs, in which the aforementioned problems are 
APX-hard~\cite{chlebik2002approx-Steiner,engebretsen2001approx-TSP, 
guha1999approx-facility}.

Technically, we achieve the above results by employing the powerful machinery of 
metric space embeddings~\cite{fakcharoenphol2003tight,bartal1998approximating}.
Specifically, for any $\eps>0$ we probabilistically compute a low-treewidth 
graph $H$ on the same vertex set as the input graph $G$ such that the 
shortest-path distance between any two vertices in $H$ is lower bounded by their 
distance in~$G$, and, in expectation, upper bounded by $1+\eps$ times their 
distance in~$G$. The latter factor by which the distances are bounded from 
above is called the \emph{distortion} or \emph{stretch} of the embedding $H$ 
(see \autoref{sec:talwar} for formal definitions). The following is the main 
result of this paper, where the {\em aspect ratio} is the maximum distance of 
any two vertices divided by the minimum distance between any vertices.

\begin{thm}\label{thm:main}
Let $G$ be a graph with highway dimension $k$ of violation $\lambda>0$, 
and aspect ratio $\alpha$. For any $\eps>0$, there is a
polynomial-time computable probabilistic embedding $H$ of $G$ with 
treewidth $(\log\alpha)^{O\left(\log^2(\frac{k}{\eps\lambda})/\lambda\right)}$ 
and expected distortion~$1+\eps$.
\end{thm}

Low highway dimension graphs do not exclude fixed-size minors and therefore do 
not have low tree\-width~\cite{robertson1986minorsII}: the complete graph on 
vertices $\{1,\ldots, n\}$ where each edge $\{i,j\}$ with $i>j$ has 
length~$c^i$, has highway dimension~$1$. The example also shows that the aspect 
ratio of a low-highway dimension graph can be exponential. Using standard 
techniques, we will show that the aspect ratio may be assumed to be polynomial 
for our considered problems when aiming for $1+\eps$ approximations. Existing 
algorithms for bounded treewidth graphs~\cite{ageev1992facility,bateni2011prize} 
then imply QPTASs on graphs with constant highway dimension (see 
\autoref{sec:ptas} for more details).

While Travelling Salesman, Facility Location, and Steiner Tree are
APX-hard in general graphs, improved algorithms are known in 
special cases. For example, polynomial time approximation schemes (PTASs) for 
all three of these problems are known if the input metric is
low-dimensional Euclidean or  
planar~\cite{arora1998TSP,arora1998k-median,mitchell1999guillotine,
bateni2011prize,borradaile2007Steiner,klein2008TSP,ageev-facility}. 
\citet{talwar2004bypassing} also showed that the work in 
\cite{arora1998TSP,arora1998k-median,mitchell1999guillotine} extends (albeit 
with quasi-polynomial running time) to low {\em doubling dimension} metrics. 
\citet{bartal2012traveling} later presented a PTAS for Travelling Salesman 
instances in this class.  

The concept of doubling dimension was studied by \citet{gupta2003bounded}, and 
captures metrics that have {\em bounded growth}.
Formally, a metric space $(X,\dist)$ has doubling dimension~$d$ if $d$
is the smallest number such that every ball
of radius $2r$
is contained in the union of $2^d$ balls of radius $r$.  The
class of constant doubling dimension metrics strictly generalizes that
of Euclidean metrics in constant dimensions. Doubling dimension and
highway dimension (as defined here) are incomparable metric
parameters, however: \citet{abraham2010highway} noted that grids have
doubling dimension~$2$ but highway dimension $\Theta(\sqrt{n})$, while
stars have doubling dimension $\Theta(\log n)$ and highway
dimension~$1$.

We briefly note here that there are alternative definitions of highway 
dimension (see \autoref{sec:alt-defs} for a detailed discussion). In particular, 
the more restrictive definition 
in~\cite{abraham2010highway2} {\em implies} low doubling-dimension, and hence 
\citet{talwar2004bypassing} readily yields a QPTAS for the optimization problems 
we study. Our choice of definition is deliberate, however, and motivated by the 
fact that \autoref{dfn:hd} captures natural transportation networks that the 
more restrictive definition does not. For instance, typical \emph{hub-and-spoke} 
networks used in air traffic models are non-planar and have high doubling 
dimension, since they feature high-degree stars. This immediately renders them 
incompatible with the highway dimension definition 
in~\cite{abraham2010highway2}. Nevertheless they have low highway dimension by 
\autoref{dfn:hd}, since the airports act as hubs, which become sparser with 
growing scales as longer routes tend to be serviced by bigger airports. We also 
prove in \autoref{sec:alt-defs} that our definition is a strict generalization 
of the one in~\cite{abraham2010highway2}: any graph with highway dimension $k$ 
according to~\cite{abraham2010highway2} has highway dimension $O(k^2)$ according 
to \autoref{dfn:hd}, while a corresponding lower bound is not possible in 
general.

Our results not only provide further evidence that the highway 
dimension parameter is useful in characterizing the complexity of graph 
theoretic problems in combinatorial optimization. Importantly, they also show 
that the geometric toolkit 
of~\cite{arora1998TSP,mitchell1999guillotine,arora1998k-median} extends beyond 
the class of low doubling dimension metrics, since the proof of 
\autoref{thm:main} heavily relies on the embedding techniques proposed 
in~\cite{talwar2004bypassing}.

\subsection{Our techniques}

The embedding constructed in the proof of \autoref{thm:main} heavily
relies on previous work by \citet{talwar2004bypassing} but needs many
non-trivial new ideas, a few of which we sketch here. 

Talwar's embedding algorithm first computes a so called {\em
  split-tree decomposition}, a certain laminar family of subsets of
the set $X$ of points underlying the given metric space.  Initially,
this family contains just one element, the set $X$ itself. In each
step, the algorithm picks a non-singleton leaf $C$ of the family,
partitions it into sets $C_1, \ldots, C_q$ of random diameter roughly
half of that of~$C$, and adds these to the family.  The algorithm
continues until all the leaves in the family are singletons. An element $C$
of the computed decomposition is commonly referred to as a {\em cluster}.

Each cluster $C$ of the split-tree decomposition is associated with a set
of \emph{net points}; net points are well spaced in $C$, and each
point in $C$ is close to at least one of these. 
For each cluster, only the edges between the net
points of its child clusters are kept to form the embedding. The
shortest path between two points can then be approximated by a path
that exits each cluster only via the net points. The error introduced
due to the shifting of points on a path to net points, as well as the
total distortion, can be bounded as the sum of errors over all levels
of the split-tree decomposition. In the tree decomposition (see
\autoref{sec:talwar} for formal definitions) of the resulting
embedding, each bag corresponds to a cluster and consists of the net
points of its child clusters. Using the bounded doubling dimension
assumption, the number of child clusters and number of net points per
cluster can be bounded by constants depending on the doubling
dimension and the desired stretch. This in turn bounds the embedding's
treewidth.

\begin{wrapfigure}[13]{R}{0.35\textwidth}
\vspace{-3mm}
\centering{\includegraphics[width=0.34\textwidth]{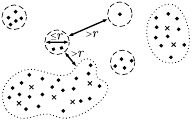}}
\caption{\label{fig:towns} The sprawl (enclosed by dotted lines) contains 
vertices close to hubs (crosses). Each town (dashed circles) has small diameter 
and is far from other vertices.}
\end{wrapfigure}

We want to construct a similar recursive decomposition for metrics 
with low highway dimension, but this turns out to be non-trivial. 
In order to obtain a decomposition we observe that the hubs in the shortest 
path cover induce a natural clustering of the vertices in $G$ for any scale~$r$ 
(see \autoref{fig:towns}). Each vertex $v \in V$ whose distance from any hub is 
larger than $2r$ is said to belong to a \emph{town} that is contained in the 
ball of radius $r$ centered at $v$. All vertices that are not part of a town 
(and hence at distance no more than $2r$ from some hub) are said to be part of 
the \emph{sprawl}. We will show that towns are nicely separated from other towns 
and the sprawl and that the degree of separation is highly sensitive to the 
choice of $c$ in \autoref{dfn:hd}. It turns out that choosing $c=4$ yields a 
separation that is just barely too small.

Based on this clustering, we compute a hierarchical decomposition of
the graph that we call the \emph{towns decomposition}. It is a laminar
family of towns and recursively separates the graph into towns of
decreasing scales, and our embedding is computed recursively on this
decomposition. The towns decomposition is analogous to the quad-tree
decomposition in PTASs for Euclidean
metrics~\cite{Arora96polynomialtime, arora1998k-median, arora1998TSP,
  arora2003survey} or the split-tree decomposition for low doubling
dimension metrics~\cite{talwar2004bypassing}, though the particulars
differ greatly. At a high level, towns look similar to clusters in
Talwar's split-tree decomposition. However, while in Talwar's
split-tree decomposition, clusters have a relatively small number of
child clusters, towns can contain a very large number of child towns. As it 
turns out, however, these child towns are connected through hubs of higher 
scales, which can be chosen in a way such that they have bounded doubling 
dimension. We can therefore apply Talwar's decomposition technique to these 
connecting hubs. We then recursively construct a low treewidth embedding for 
each child town and attach these embeddings to the embedding of the connecting 
hubs. The details are described in \autoref{sec:construction}.

The most intricate part of our result is to prove low doubling dimension of 
these ``connecting hubs'', which are chosen as follows. We prove that to 
preserve all distances within a town $T$ it suffices to connect embeddings of 
$T$'s child towns in the towns decomposition via a carefully chosen set of 
so-called \emph{core hubs} within~$T$. To prove low doubling dimension, the 
general idea is to rely on the local sparsity of the shortest path covers (see 
\autoref{fig:core-hubs}): by definition, the core hubs lie in the sprawls of 
various scales, and for scale $r$ the sprawl can be covered by balls of radius 
$2r$ around the hubs of the shortest path cover. In a low highway dimension 
graph, any ball~$B$ of radius $cr/2$ contains only a small number of hubs. 
Hence, to bound the doubling dimension, we attempt to use these hubs as centers 
of balls of smaller radius to cover the core hubs. These balls have radius 
$2r<cr/2$, and hence this scheme can be applied recursively in order to cover 
the core hubs in~$B$ with balls of half the radius. Several issues arise with 
this approach though. For instance, part of the sprawl for scale $r$ in $B$ 
might be covered by balls centered at hubs outside of~$B$. However a key insight 
of our work is that in fact the number of hubs in the \emph{vicinity} of a ball 
is also bounded when using \autoref{dfn:hd} for the highway dimension (see 
\autoref{lem:hub_bound}).

\begin{wrapfigure}[13]{R}{0.35\textwidth}
\vspace{-8mm}
\centering{\includegraphics[width=0.34\textwidth]{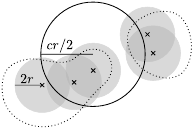}}
\caption{\label{fig:core-hubs} The sprawl (enclosed by dotted lines) 
intersecting a ball $B$ of radius~$cr/2$ (black) can be recursively covered by 
balls of radius $2r$ (grey) centered at hubs on scale $r$ (crosses). For this 
the number of hubs in the vicinity of~$B$ needs to be bounded.}
\end{wrapfigure}

Another obstacle when trying to bound the doubling dimension of the core 
hubs is that, unlike the nets in Talwar's split-tree decomposition, the hubs do 
not form a hierarchy, i.e., a hub at some scale may not be a hub at a lower 
scale. Nevertheless, we show that core hubs at different scales can be {\em 
aligned}: they can be shifted slightly in order to obtain a nested structure. We 
are able to show that this alignment process does not affect the target stretch 
of our embedding and, most importantly, ensures that the resulting set of {\em 
approximate core hubs} within $T$ has small doubling dimension. We may thus 
apply Talwar's~\cite{talwar2004bypassing} embedding of low doubling dimension 
metrics into bounded treewidth graphs to the approximate core hubs.

\subsection{Related work}

The highway dimension concept was introduced by \citet{abraham2010highway} who showed 
that the efficiency of certain shortest-path heuristics can be explained with 
this parameter. Follow-up papers~\cite{abraham2011vc,abraham2010highway2} 
introduced alternative definitions and showed that it is possible to 
approximate the highway dimension~$k$ within an $O(\log k)$ factor assuming 
that shortest paths are unique. For the $p$-Center problem the embedding 
techniques given in this paper are not applicable since the objective function 
is non-linear. Instead, in~\cite{feldmann15} a parameterized approximation for 
this problem on low highway dimension graphs is presented. 
\citet{bauer2013search} show that for any graph $G$ there exist edge 
lengths such that the highway dimension is $\Omega(\mathrm{pw}(G)/\log n)$, 
where $\mathrm{pw}(G)$ is the pathwidth of~$G$. Also \citet{kosowski2017beyond} 
consider the highway dimension and compare it to the related skeleton dimension.

In the seminal work of \citet{bartal1998approximating,bartal1996probabilistic} 
it was shown that any graph can be embedded into a distribution over trees 
with an expected polylogarithmic stretch. The stretch bound was later improved 
to $O(\log n)$ by \citet{fakcharoenphol2003tight}, which is the best possible. 
These techniques led to the embedding of low doubling dimension metrics into 
bounded treewidth graphs by~\citet{talwar2004bypassing}, which forms a major 
ingredient in our result. Another generalization is that of 
\citet{hubert2012global}, who showed how to embed a metric of low correlation 
dimension into a metric of bounded treewidth. It it worth noting that the 
highway dimension cannot be bounded in terms of the correlation dimension (due 
to the complete graph example described above). In terms of lower bounds, 
there are graphs~\cite{carroll2004lower,chakrabarti2008local} with 
treewidth~$t$, which cannot be embedded into distributions over graphs excluding 
minors of size~$t-1$, without incurring an expected stretch of $\Omega(\log n)$. 
The authors also show that embeddings of planar graphs into bounded treewidth 
graphs must incur logarithmic distortions.

\section{Embeddings for low doubling dimension metrics}
\label{sec:talwar}

Next we formally define the treewidth and summarize the properties of 
Talwar's~\cite{talwar2004bypassing} embedding for low doubling dimension metrics 
that we require for our construction. More details will be given in 
\autoref{sec:stretch}, which are needed for the analysis of the 
stretch of our embedding.

Let $G = (V,E)$ be a graph. For $u, v \in V$ we denote the length of the 
shortest path between $u$ and~$v$  by $\dist(u,v)$ and the distance between two 
sets $S, T \subset V$ by $\dist(S,T) = \min_{u \in S, v \in T} \dist(u,v)$. If 
the metric used for distances is ambiguous we specify the graph in the 
subscript, such as $\dist_G(u,v)$ or $\dist_H(u,v)$.  The diameter 
$\diam(\cdot)$ of a graph or set of vertices is the maximum distance between any 
two vertices. The \emph{treewidth} of a graph measures how close the graph is 
from being a tree. A {\em tree decomposition} of $G$ consists of a tree $T$ 
whose vertices are labelled by subsets of $V$ that are commonly referred to as 
{\em bags}. We will often identify the bags with the vertices of the tree and 
talk about a ``tree of bags''.
Bags satisfy certain structural properties as is formalized in the following 
definition. 

\begin{dfn}\label{dfn:treewidth}
A \emph{tree decomposition} $D$ of a graph $G=(V,E)$ is a
tree $T$
each of whose vertices $v$ are labelled by a bag $b_v \subseteq V$ of
vertices of $G$. We require the following properties:
 \begin{enumerate}[(a)]
\item\label{item:tw-union} $\bigcup_{v \in V(T)} b_v = V$, 
\item\label{item:tw-edges} for every edge $\{u,w\}\in E$ there is a
  vertex $v \in V(T)$ such that $b_v$ contains both $u$ and $w$, and
\item\label{item:tw-vertices} for every $v\in V$ the set $\{u \in V(T)
  \,:\, v \in b_u\}$ induces a connected subtree of $T$.
 \end{enumerate}
The \emph{width} of the tree decomposition is $\max\{|b_v|-1\,:\, v
\in V(T)\}$. The \emph{treewidth} of a graph $G$ is the minimum width 
among all tree decompositions for~$G$.
\end{dfn}

To construct our embedding we will mainly focus on the shortest path
metric of the  graph~$G$. We let the distance function of
every considered metric be the function $\dist(\cdot,\cdot)$
of the underlying graph. Though the treewidth is a property of a graph's
edge set, whereas doubling dimension is a property of the metric
it defines, \citet{talwar2004bypassing} shows that low doubling
dimension graphs can be approximated to within $1+\eps$ by bounded
treewidth graphs. Formally this means the following.

\begin{dfn}\label{dfn:embedding}
Let $(X,\dist)$ be a metric, and $\mc{D}$ be a distribution over 
metrics~$(X,\dist')$. If for all $x,y \in X$, $\dist(x,y) \le \dist'(x,y)$ for 
each $\dist' \in \mc{D}$, and $\E_{\dist' \in \mc{D}}[\dist'(x,y)] \le a\cdot 
\dist(x,y)$, then $\mc{D}$ is an \emph{embedding} with (expected) 
\emph{stretch} or \emph{distortion} $a$. If every $\dist'\in\mc{D}$ is the 
shortest path metric of some graph class $\mc{G}$, then $\mc{D}$ is a 
\emph{(probabilistic) embedding into} $\mc{G}$.
\end{dfn}

The main result of \citet{talwar2004bypassing} that we use for our embedding of 
low highway dimension graphs into bounded treewidth graphs, is the following.

\begin{thm}[\cite{talwar2004bypassing}]
\label{thm:Talwar-light}
Let $(X,\dist)$ be a metric with doubling dimension $d$ and aspect ratio 
$\alpha$. For any $\eps>0$, there is a polynomial-time computable probabilistic 
embedding $H$ of $(X,\dist)$ with treewidth $(d\log(\alpha)/\eps)^{O(d)}$ and 
expected distortion~$1+\eps$.
\end{thm}

As described in the introduction,
Talwar's embedding employs a randomized \emph{split-tree} decomposition, which 
is a hierarchical decomposition of the vertices $X$ of a metric into 
\emph{clusters} of smaller and smaller diameter. A cluster is a subset of~$X$, 
which is partitioned into clusters of at most half the diameter on the next 
lower level, so that the highest cluster is $X$ itself and the lowest ones are 
individual vertices. The geometrically decreasing diameters of the levels are 
set according to a random variable. Each level of this hierarchy is associated 
with an index. Our construction of the embedding for low highway dimension 
graphs also has levels associated with indices, but these have different growth 
rates. To avoid confusion we will denote the levels of Talwar's split-tree 
decomposition with indices $\bar i,\bar j$, etc., and ours with indices $i,j$ 
etc. 

The tree decomposition constructed from the split-tree has a bag for each 
cluster. The tree on the bags exactly corresponds to the split-tree. Each bag 
contains a coarse set of points of the cluster. More concretely it contains a 
\emph{net}, defined as follows.

\begin{dfn}
\label{dfn:netcover}
For a metric $(X,\dist)$, a subset $Y\subseteq X$ is called a 
\emph{$\delta$-cover} if for every $u\in X$ there is a $v\in Y$ such that 
$\dist(u,v)\leq \delta$. A \emph{$\delta$-net} is a $\delta$-cover with the 
additional property that $\dist(u,v)>\delta$ for all vertices $u,v\in Y$. 
\end{dfn}
For a 
cluster $C$ on level $\bar i$ the corresponding bag contains a 
$\Theta(\eps 2^{\bar{i}}/(d\log\alpha))$-net of $C$.
For every bag $b$ the graph embedding contains a complete graph on the nodes in 
$b$ with edge lengths corresponding to distances in the metric. The net in each 
bag serves as a set of \emph{portals}, through which connections leaving the 
cluster are routed, analogous to those in~\cite{arora2003survey}.

\section{Properties of low highway dimension graphs}
\label{sec:properties}

We assume w.l.o.g.\ that every shortest path in our input graph is
unique by slightly perturbing edge lengths. This allows us to
compute locally \mbox{$O(k\log k)$-sparse} shortest path covers in
polynomial time~\cite{abraham2011vc} (or locally $k$-sparse covers in
time $n^{O(k)}$). We show in \autoref{sec:alt-defs} that computing the
highway dimension is NP-hard even for graphs with unit edge lengths,
so in general approximations are needed.

An important observation is that the vertices of low highway dimension 
graphs are grouped together in all regions that are far from the hubs. 
This gives rise to our main observation on the structure of low highway 
dimension graphs, as summarized in the following definition: for any scale the 
vertices are partitioned into one \emph{sprawl} and several \emph{towns} with 
large separations in between.

\begin{dfn}\label{dfn:towns}
Given a shortest path cover $\spc(r)$ for scale $r$, for any vertex $v\in V$ 
such that $\dist(v,\spc(r)) > 2r$, we call the set $T=\{u \in V | \dist(u,v) 
\le r\}$ a \emph{town} for scale $r$. The \emph{sprawl} for scale $r$ is the set 
of all vertices that are not in towns.
\end{dfn}

Note that the vertices of the sprawl are at most $2r$ away from a hub, but 
there can be vertices in towns that are closer than $2r$ to some hub, as long as 
the town has some other vertex that is farther away. Note also that the towns 
are defined with respect to a shortest path cover $\spc(r)$, and using two 
distinct shortest path covers will not always result in the same set of towns. 
We will fix an inclusion-wise minimal shortest path cover $\spc(r)$ for any scale $r$ and only 
consider towns with respect to this cover. We summarize the basic properties of 
towns below.

\begin{lem}
\label{lem:townproperties}
Let $T$ be a town of scale $r$. Then $\diam(T) \le r$ and $\dist(T,V\setminus T) 
> r$. For any vertex $v$ of the sprawl of scale $r$, $\dist(v,\spc(r))\leq 2r$.
\end{lem}
\begin{proof}
The bound on the distance from any vertex of the sprawl to the nearest hub 
follows immediately from the definition of the towns. To prove that the diameter 
of a town $T$ is at most $r$, assume there are vertices $u,w\in T$ such that 
$\dist(u,w)>r$. By \autoref{dfn:towns} we know there is a vertex $v\in T$ such 
that $\dist(u,v)\leq r$ and $\dist(w,v)\leq r$, so that $\dist(u,v)\leq 2r$. 
This means that the length of the shortest path between $u$ and $w$ lies in the 
interval $(r,cr/2]$, as by \autoref{dfn:hd} the constant $c$ defining $\spc(r)$ 
is at least~$4$. In particular, there is a hub $h\in\spc(r)$ that lies on this 
shortest path. Assume w.l.o.g.\ that $h$ is closer to $w$ than to $u$, so that 
$\dist(h,w)\leq r$. But then, $\dist(h,v)\leq\dist(h,w)+\dist(w,v)\leq 2r$, 
which contradicts $\dist(v,\spc(r))> 2r$.

Similarly, we can prove that the distance of any vertex $u$ of a town $T$ to any 
vertex $w$ outside of $T$ is more than $r$. Consider again the vertex $v\in T$ 
given by \autoref{dfn:towns}, for which $\dist(u,v)\leq r$, $\dist(w,v)>r$, and 
$\dist(v,\spc(r))>2r$. If we assume that $\dist(w,u)\leq r$, then from the first 
distance bound for $u$ and~$v$ we get $\dist(w,v)\leq 2r$. Together with 
$\dist(w,v)>r$, this means that the length of the shortest path $P$ between $w$ 
and $v$ lies in the interval $(r,cr/2]$, as by \autoref{dfn:hd} $c\geq 4$. Hence 
there is a hub $h\in\spc(r)$ on $P$ that is at most as far from $v$ as $w$ is, 
i.e.\ $\dist(v,h)\leq 2r$. However this contradicts $\dist(v,\spc(r))>2r$.
\end{proof}

We will exploit the structure given by \autoref{lem:townproperties} for growing 
scales to construct our embedding. 
More concretely, we will consider scales $r_i=(c/4)^i$ for values 
$i\in\mathbb{N}_0$ and call~$i$ the \emph{level} of the sprawl, towns, and 
shortest path cover of scale $r_i$. We choose our scales in this way since 
$2r_i=cr_{i-1}/2$. As a consequence, a ball of radius $2r_i$ around a hub of 
level $i$ that covers part of the sprawl contains at most $s$ hubs of the next 
lower level $i-1$ if the shortest path covers are locally $s$-sparse. We will
exploit this in our analysis in order to bound the treewidth of our embedding.

Note that the scales are monotonically non-increasing if we choose
$c \leq 4$. As we shall see, positive scale-growth is essential,
however, for our algorithm as it allows us to argue that any two
disjoint towns are sufficiently separated.

Throughout this paper we will assume that the shortest path covers are 
inclusion-wise minimal. By scaling we can assume that the shortest distance 
between any two vertices is slightly more than~$c/2$. Hence 
$\spc(r_0)=\emptyset$ since there are no paths of length in~$(r_0,cr_0/2]$. In 
particular this means that on level~$0$ there is no sprawl, and each vertex 
forms a singleton town. The highest level we consider is 
\mbox{$m=\lceil\log_{c/4}\diam(G)\rceil$}. At this level $\spc(r_m)=\emptyset$ 
and hence the whole vertex set $V$ of the graph is a town.

We show next that towns of different levels form a laminar family $\mc{T}$. Due 
to this laminar structure of towns we will use tree terminology such as 
\emph{parents, children, siblings, ancestors}, and \emph{descendants} of towns 
in~$\mc{T}$. Note that these family relations are with respect to the 
laminarity of $\mc{T}$ and not the levels on which towns exist. The \emph{root} 
of the laminar family is the highest level town $V$.

\begin{lem}\label{lem:laminar-towns}
Given a graph $G$, the set
$
\mc{T}:=\{T\subseteq V\mid T\text{ is a town on level } i\in \mathbb{N}_0\}
$
forms a laminar family. Furthermore, any town $T\in\mc{T}$ on level $i$ either 
has 0 or at least 2 child towns, and in the latter case these are towns 
on levels below~$i$.
\end{lem}

\begin{proof}
Assume $\mc{T}$ is not laminar, in which case there are two towns $T_1$ and 
$T_2$ in $\mc{T}$ that cross, i.e., all of the sets $T_1\cap T_2$, 
$T_1\setminus 
T_2$, and $T_2\setminus T_1$ are non-empty. Assume that $T_1$ is a town of 
level 
$i$, while $T_2$ is a town of level $j\geq i$. Let $v $ and $w$ be two vertices 
of $T_1$ such that $v\in T_2$ but $w\notin T_2$. By 
\autoref{lem:townproperties}, $\dist(v,w)\leq \diam(T_1)\leq r_i$ and 
$\dist(v,w)\geq \dist(T_2,V\setminus T_2) > r_j\geq r_i$---a contradiction.

For the second part, let $T$ be a town in the set $\mc{T}$ with a child $T'$. 
Note that $T\setminus T'\neq\emptyset$, while every vertex is a town on 
level~$0$. So there must be another town that is a child of $T$. Now assume 
there is a town $T$ on level $i$ with a child town $T'$ on level~$j\geq i$. By 
\autoref{lem:townproperties}, the diameter of $T$ is at most $r_i$, and
 any other child town of $T$ must be at distance more than 
$r_j\geq r_i$ from~$T'$. This would mean that $T$ only has one child 
town---a contradiction.
\end{proof}

The above lemma proves that the following procedure has a well-defined output:
starting with a town $T$ on some level $i$, repeatedly remove child towns on 
level $i-1$ until only the sprawl remains. Continue by removing all towns on
level $i-2$, $i-3$, etc.~from the remaining nodes until all nodes have been 
removed. Then recurse on each of the removed child towns.

Starting the decomposition with town $G$ on level $\log_{c/4} \diam(G)$,
we refer to the resulting laminar family $\mc{T}$ as the \emph{towns decomposition} of 
$G$. Note that $\mc{T}$ partitions every town $T \in \mc{T}$, and 
although $T$ appears once in $\mc{T}$, $T$ can 
be a town on multiple levels of the shortest path covers, if it is a town with 
respect to both $\spc(r_i)$ and $\spc(r_{i+1})$.
From now on we will consider the graph metric $(V,\dist_G)$ induced by $G$ 
instead of $G$ itself. All properties of towns and sprawl, such as given by 
\autoref{lem:townproperties} and~\ref{lem:laminar-towns}, are still valid in 
the metric. 
\section{Constructing the embedding}
\label{sec:construction}

We now describe our algorithm in more detail. 
PTASs for Euclidean and low doubling 
metrics~\cite{arora2003survey,talwar2004bypassing} use hierarchical 
decompositions coupled with a small number of ``portal'' nodes:
paths leaving a cluster in the decomposition must do so via an appropriate portal,
resulting in a small ``interface'' between distinct clusters in the decomposition.
Intuitively, the hubs are natural choices for portals, since long paths through 
some ball must pass through a hub. However problems crop up almost immediately 
because hubs are not guaranteed to be well-spaced or consistent between levels, 
and although all long paths through a ball may be hit by portals, there may be 
many short paths that go nowhere near one.

We overcome these difficulties by exploiting the properties of the towns 
decomposition. 
\autoref{lem:townproperties} guarantees that towns are isolated from both each 
other and the sprawl. Consequently, any approximate shortest path between nodes 
in a town must remain within that town.
The embedding is constructed recursively on the metric using the structure of 
the towns decomposition $\mc{T}$. That is, for a town $T\in\mc{T}$ we assume 
that we have already computed an embedding (and accompanying tree decomposition) 
with expected stretch $1+\eps$ for each child town of~$T$. We then connect 
these embeddings so that distances between them are preserved within a~$1+\eps$ 
factor in expectation. This gives an embedding for $T$ and, since $V$ itself is 
the root of the towns decomposition, eventually yields an embedding for $G$. 

The key insight that lets us connect the child towns of $T$ is that there 
exists a set of so-called \emph{approximate core hubs}~$X_T$ in~$T$ with low 
doubling dimension that can serve as the crossroads through which child towns 
connect. We will compute a low-treewidth embedding of the set~$X_T$ based on 
\autoref{thm:Talwar-light} and connect the embeddings of the child towns to it. 
In particular, for every child town~$T'$ we will identify a bag $b$ of the tree 
decomposition of $X_T$ containing hubs that are close to~$T'$. We call~$b$ the 
\emph{connecting bag} of $T'$. 
The embedding of~$T$ is constructed by connecting every vertex in each child 
town to every hub in the corresponding connecting bag. As we show in 
\autoref{sec:stretch}, this means that short connections between child towns can 
be routed directly through hubs in the connecting bags. Long connections on the 
other hand can be routed through the embedding of the core hubs $X_T$ at only a 
small overhead.

\begin{wrapfigure}[18]{R}{0.3\textwidth}
\centering{\includegraphics[width=0.25\textwidth]{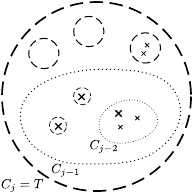}}
\caption{\label{fig:core} The cores of three different levels of town $T$ 
(enclosed by dotted lines for levels $i<j$). Note that some hubs of level $j-2$ 
(small crosses) lie in towns of level $j-1$ (larger dashed circles), and these 
are not core hubs.}
\end{wrapfigure}

The tree decomposition for $T$ is constructed by connecting each tree 
decomposition $D_{T'}$ for a child town $T'$ to the corresponding connecting bag 
$b$ of the tree decomposition $D_X$ for the hubs in~$X_T$ (lines~\ref{decomp1} 
to \ref{decomp2} in \autoref{alg_embedding}). Even though this yields a tree of 
bags containing all vertices of the town $T$, properties~\eqref{item:tw-edges} 
and~\eqref{item:tw-vertices} of \autoref{dfn:treewidth} might be violated by 
this initial attempt. As we will show in \autoref{sec:treewidth}, we need to 
make two modifications to the bags: first we need to add all vertices of $b$ to 
each bag of~$D_{T'}$. Since the treewidth of $D_X$ is bounded by 
\autoref{thm:Talwar-light}, this does not increase the sizes of bags by much. 
Second, we also need to add all hubs of $X_T$ within the child town $T'$ to each 
bag of $D_{T'}$, as well as to $b$ and all descendants of $b$ in~$D_X$. To bound 
the growth of the bags in this step, we need to bound the number of hubs in 
$X_T$ in a child town~$T'$, which we do in \autoref{sec:treewidth}.

The set $X_T$ is an approximate hub set of $T$. To define the set properly we 
need some additional insights on the structure of hubs of different levels 
in~$T$. The \emph{core} of $T$ is the intersection of sprawls formed by removing 
all child towns of $T$ above a given level (c.f.~\autoref{fig:core}):

\begin{dfn}\label{dfn:core}
Let $T\in\mc{T}$ be a town on level $j$, and let $S_i$ be the sprawl of $V$ on 
level~$i \leq j$. The \emph{core~$C_i$ of $T$ on level $i$} is inductively 
defined as 
follows: $C_{j}=T$, and $C_i=S_i\cap C_{i+1}$ for $i\leq j-1$. The \emph{core 
hubs of $T$} are given by the set $\bigcup_{i=1}^{j-1} C_i\cap \spc(r_i)$.
\end{dfn}

By this definition a town $T$ on level $j$ can be partitioned into its core on 
level~$i$ and its child towns on levels $\{i,\ldots,j-1\}$. Observe also that 
the set system $\{C_i\}_{i=0}^{j}$ given by the cores forms a chain, 
i.e.~$C_{i-1}\subseteq C_i$. Intuitively, the core hubs should have low doubling 
dimension: if the shortest path covers are locally $s$-sparse, then in a ball 
around a hub at level $i$ there will be at most $s$ hubs in that ball on 
level~$i-1$, and the balls of half the radius around these hubs cover the core 
on that level (cf.~\autoref{fig:core-hubs}). In fact one can show that the 
doubling dimension of the core hubs is fairly small but unfortunately not small 
enough for our purposes. In particular, we need the doubling dimension to be 
independent of the aspect ratio $\alpha$ of the metric. To circumvent this 
issue, roughly speaking, we shift each core hub so that it overlaps with lower 
level core hubs if possible, making the hubs nested to some degree. However, in 
order to preserve distances we will only shift them by at most an $\eps$ 
fraction. This shifting produces the set $X_T$ of \emph{approximate core hubs} 
of~$T$, which we use to construct our core embedding. Note that we do not use 
the approximate core hubs $X_T$ to define our towns decomposition, only to 
produce a low-treewidth core embedding (see lines~\ref{core-hubs} 
and~\ref{towns-decomp} in \autoref{alg_embedding}). We rely on the following 
non-trivial properties, which require an intricate proof provided in 
\autoref{sec:dd}.

\begin{thm}\label{thm:doubling_dim}
Let $\mc{T}$ be a towns decomposition of a graph of highway dimension~$k$, given 
by locally $s$-sparse shortest path covers on all levels with violation 
$\lambda>0$. For any town $T\in\mc{T}$ of a level~$j$ there exists a 
polynomially computable set of approximate core hubs $X_T\subseteq T$ such that
\begin{itemize}
\item for any core hub $h\in C_i\cap\spc(r_i)$ of $T$ on level 
$i\in\{1,\ldots,j-1\}$, there is a vertex $h'\in X_T$ with $\dist_G(h,h')\leq 
\eps r_i$, and 
\item the doubling dimension of $X_T$ is
$d = O(\log(\frac{ks\log(1/\eps)}{\lambda})/\lambda)$.
\end{itemize}
\end{thm}

From now on, we use $d$ to denote the above doubling dimension bound for $X_T$. 
Our algorithm computes the low-treewidth embedding $H_T$ of $T$ by explicitly 
computing its tree-decomposition $D_T$. The latter is constructed by connecting 
the recursively computed tree decompositions $D_{T'}$ for child towns $T'$ of 
$T$ to the tree decomposition $D_X$ of an embedding $H_X$ for the metric induced 
by the approximate core hubs $X_T$. For this to work we need to make sure that 
the approximate core hubs contained in the same child town $T'$ do not end up in 
different bags in the tree decomposition $D_T$ of $H_T$. Our solution is to pick 
a representative core hub for each child town~$T'$. Specifically, let 
$Y_T\subseteq X_T$ contain one arbitrary approximate core hub for each child 
town $T'$ of $T$ for which $T'\cap X_T\neq\emptyset$. We say that a vertex $v\in 
Y_T$ of a child town $T'$ \emph{represents} the nodes in $X_T\cap T'$ (including 
$v$ itself). The sub-metric $Y_T$ of $X_T$ inherits the doubling dimension bound 
of \autoref{thm:doubling_dim}, since the doubling dimension of any sub-metric is 
at most twice the doubling dimension of the original metric. This was already 
noted in~\cite{gupta2003bounded}, and we give a formal proof of this fact in the 
following. We state this observation slightly more general than we need it here, 
as we will reuse it in \autoref{sec:dd}: in the next lemma the metric $Z$ is not 
required to have bounded doubling dimension, but the premise is clearly 
fulfilled if it does.

\begin{lem}\label{lem:set-d-dim}
Let $(Z,\dist)$ be a metric and $Z'\subseteq Z$. If for every ball 
$B_{2r}(v)\subseteq Z$ of radius $2r$ there are at most $2^\delta$ balls 
$B_r(u_i)\subseteq Z$, with centers $u_i$ and each with radius $r$, such that 
their union contains all vertices in $B_{2r}(v) \cap Z'$, then the doubling 
dimension of $(Z',\dist)$ is at most~$2\delta$.
\end{lem}
\begin{proof}
Any ball in $(Z',\dist)$ corresponds to a ball in $(Z,\dist)$ with a center 
vertex in $Z'$. Pick a ball $B_{2r}(v)\subseteq Z$ with radius $2r$ and $v\in 
Z'$. For each of the $2^\delta$ balls $B_r(u_i)$ that exist for $B_{2r}(v)$, 
there again are at most $2^\delta$ balls $B_{r/2}(w_{ij}) \subseteq Z$ with 
radius $r/2$ whose union contains $B_r(u_i)\cap Z'$. Pick any vertex~$w_{ij}' 
\in Z'$ (if any) in such a ball $B_{r/2}(w_{ij})$ and consider the ball 
$B_r(w_{ij}')$ of double the radius. This ball must contain~$B_{r/2}(w_{ij})$. 
Doing this for all such balls $B_{r/2}(w_{ij})$ gives at most $2^{2\delta}$ 
balls, each with a center vertex in $Z'$, such that their union covers 
$B_{2r}(v)\cap Z'$. Hence the ball $B_{2r}(v)\cap Z'$ in $(Z',\dist)$ is covered 
by at most $2^{2\delta}$ balls in $(Z',\dist)$ by intersecting each of these 
balls in $(Z,\dist)$ with~$Z'$.
\end{proof}

\begin{algorithm2e}[b!]
\caption{Compute embedding $H$ with tree decomposition $D_H$ of graph $G$}
\label{alg_embedding}

\LinesNumbered
\DontPrintSemicolon
\SetAlgoVlined

\For {$i = 0,\ldots, \lceil\log_{c/4} 
\diam(G)\rceil$} 
{
$\spc(r_i)\leftarrow$ locally $O(k\log k)$-sparse minimal shortest path cover
\tcp{See~\cite{abraham2011vc}}
}
$\mc{T}\leftarrow$ towns decomposition based on $\spc(r_i)$\;
\label{towns-decomp}
$(H,D_H)=Embed(V, \lceil\log_{c/4} \diam(G)\rceil)$ \tcp{Recursively compute 
embedding $H$ with tree decomposition $D_H$}

\BlankLine
\SetKwProg{Fn}{function}{}{}
\Fn{$Embed(T, j)$ \tcp*[h]{Low-treewidth embedding of town $T$ at level $j$}  } 
{
\lIf{$j=0$} {
\KwRet {$(T,T)$ \tcp*[h]{A town is a singleton at level $0$} }
} 

Compute approximate core hubs $X_T$ of $T$ \tcp*[h]{According to
\autoref{thm:doubling_dim}}\; \label{core-hubs}

$Towns \leftarrow \emptyset$ \tcp{Set of embeddings of child towns of $T$}

\SetKw{kwwhile}{while}
\For( \tcp*[h]{Recurse on child towns}) {$i = j-1,\ldots, 0$ 
} { 
\ForEach {child town $T'\in\mc{T}$ of $T$ on level $i$} { 

$(H_{T'},D_{T'}) \leftarrow Embed(T',i)$  \; 
\label{recurse}
Add $(H_{T'}, D_{T'} ,i)$ to $Towns$ \;
}

}

\tcp{Compute embedding $H_X$ for $X_T$ with tree decomposition $D_X$}
$Y_T \leftarrow $ one node in $X_T \cap T'$ for each child town $T'$ of $T$ for 
which $X_T \cap T'\neq\emptyset$ \; \label{hub-embedding1}
$(H_Y,D_Y) \leftarrow Talwar(Y_T,\eps')$ \tcp{Embedding of $Y_T$ with 
distortion $1+\eps'$}
$(H_X,D_X) \leftarrow$ expand each vertex in $H_Y,D_Y$ into all 
hubs it represents in $X_T$ \; \label{hub-embedding2}
$H_T \leftarrow H_X$ \tcp{Initially the embedding $H_T$ of $T$ is $H_X$}
$D_T \leftarrow D_X$ \tcp{Initially the tree decomposition $D_T$ of $T$ is 
$D_X$}
root$(D_T)\leftarrow$ root$(D_X)$ \tcp{Set the root bag of the tree 
decomposition}

\ForEach(\tcp*[h]{Join towns to $H_T$}) {$(H_{T'}, D_{T'}, i)$ in $Towns$} {
\tcp{Find the connecting bag $b$ for $T'$}
$T''\leftarrow$ closest sibling town to $T'$ in $T$\; \label{bag1}
$i\leftarrow$ level for which $\dist_G(T',T'')\in(r_i,r_{i+1}]$\;
$h\leftarrow$ closest hub in $X_T$ to $T'$\;
$\bar i\leftarrow \lceil \log_2 r_i\rceil$\;
$\bar j\leftarrow$ highest level of $D_X$\;
$C \leftarrow $ cluster containing $h$ at level $\bar l=\min\{\bar 
j,\bar i+\lceil\log_{2}(d/\eps)\rceil\}$ in split-tree of $X_T$\;
$b \leftarrow$ bag in $D_X$ corresponding to cluster $C$ \; \label{bag2}

\tcp{Connect $T'$ to $X_T$ in the embedding}
Add all vertices and edges of $H_{T'}$ to $H_T$ \; \label{connect1}
Add edge $\{u,v\}$ with length $\dist_G(u,v)$ to $H_T$ for each pair $u \in 
T'$, $v\in b$ \; \label{connect2}

\tcp{Add $D_{T'}$ to the tree decomposition $D_T$ of $H_T$}
Merge $D_{T'}$ and $D_T$ by connecting root$(D_{T'})$ with $b$\; \label{decomp1}
Add all vertices of $b$ to each bag of $D_{T'}$ \;
Add all hubs of $X_T\cap T'$ to each bag of $D_{T'}$, and also to $b$ and 
all descendants of $b$ in~$D_X$ (but not the descendants of $b$ in $D_T$ that 
are bags of some $D_{T''}$ for some child town $T''\neq T'$ of $T$)\; 
\label{decomp2}

}

\KwRet{$(H_T,D_T)$}
}

\end{algorithm2e}

By \autoref{lem:set-d-dim} the doubling dimension of $Y_T$ is at most $2d$, 
and so we can compute an embedding $H_Y$ for the metric $(Y_T,\dist_G)$ with 
bounded treewidth by \autoref{thm:Talwar-light}. Given $H_Y$ together with a 
tree decomposition $D_Y$ we convert it into an embedding $H_X$ of~$X_T$ together 
with a tree decomposition $D_X$ by replacing a vertex $v\in Y_T$ with all 
approximate core hubs that $v$ represents (see lines~\ref{hub-embedding1} 
to~\ref{hub-embedding2} in \autoref{alg_embedding}). In particular, the tree 
decomposition $D_X$ of $H_X$ is obtained from the decomposition $D_Y$ of $H_Y$ 
by replacing $v\in Y_T$ with all the hubs it represents in each bag 
containing~$v$. For every bag $b$ of $D_X$ the embedding $H_X$ contains a 
complete graph on 
the vertices of $b$, where the length of an edge $\{u,v\}$ is the distance 
$\dist_G(u,v)$ in $G$. It is easy to see that $D_X$ is a valid tree 
decomposition, i.e., it satisfies all properties of \autoref{dfn:treewidth}. We 
will show in \autoref{sec:treewidth} that the number of approximate core hubs in 
each child town is bounded, and therefore the growth of the treewidth caused by 
replacing a vertex by its represented hubs is also bounded. We also need to 
bound the extra distortion incurred by going from $H_Y$ to $H_X$ and show that a 
$1+\eps$ distortion of $H_Y$ translates into a $1+O(\eps)$ distortion of $H_X$, 
which entails reproving the relevant parts of \autoref{thm:Talwar-light}.

After computing the embedding $H_X$ for $X_T$, we connect each recursively 
computed embedding for the child towns of $T$ (line~\ref{recurse} of 
\autoref{alg_embedding}) to $H_X$ to form the final embedding $H_T$.  We need 
to argue that $H_X$ exists every time there are 
child towns to connect. From \autoref{lem:laminar-towns} we know that $T$ has 
at least two child towns if it has any. In \autoref{sec:stretch} we will show 
(in \autoref{lem:core-hub}) that there is a core hub $h$ in $T$ on any shortest 
path between a pair of children towns. By \autoref{thm:doubling_dim}, 
there is an approximate core hub in $X_T$ close to~$h$. 
Since $X_T$ is non-empty,  
$H_X$ exists. Once we compute $H_X$ we connect every vertex of a child 
town~$T'$ to all hubs in a bag $b$ of the tree decomposition $D_X$ of $H_X$. 
This bag $b$ is $\log_2(d/\eps)$ levels higher in the split-tree 
decomposition than the level corresponding to the shortest distance that needs 
to be bridged from $T'$ to any other vertex in~$T$. At the same time we will 
make sure that the net defining $b$ is fine enough so that lengths of 
connections passing through $b$ are preserved to a sufficient degree. This way, 
short connections from $T'$ to core hubs with length up to $O(1/\eps)$ times 
the separation of $T'$ are preserved in expectation by routing through the hubs 
in $b$. Connections to more distant hubs can be rerouted from a hub close to 
$T'$ 
through the embedding $H_X$ with only an~$\eps$ overhead, as we will prove in 
\autoref{sec:stretch}.

Recall that levels of the split-tree decomposition are denoted by $\bar i,\bar 
j$ etc. To determine the level of the bag~$b$, note that due to our growth rate 
of $c/4=1+\lambda/4$ of the levels (and the assumption that the 
violation~$\lambda$ is at most~$4$) the intervals $(r_i,2r_i]$ of the shortest 
path covers might overlap. As described in lines~\ref{bag1} to~\ref{bag2} of 
\autoref{alg_embedding}, let~$i$ be the level for which the distance between 
$T'$ and its closest sibling town lies in the interval $(r_i,r_{i+1}]$, and let 
$\bar i=\lceil\log_2 r_i\rceil$ be the corresponding level of the split tree 
decomposition of $D_X$. Now let $h\in X_T$ be the closest approximate core hub 
to~$T'$ (which might lie inside of $T'$). If $\bar j$ is the highest level of 
$D_X$, i.e.\ it is the level of the cluster containing all of~$X_T$, then the 
bag $b$ of the tree decomposition $D_X$ is the one on level $\bar l=\min\{\bar 
j,\bar i+\lceil\log_{2}(d/\eps)\rceil\}$ for which the corresponding cluster 
$C$ contains~$h$. All edges between vertices of $T'$ and $b$ are added to the 
embedding for~$T$ (lines~\ref{connect1} and~\ref{connect2} of 
\autoref{alg_embedding}), and we call the bag $b$ the \emph{connecting bag} 
for~$T'$.

Note that there are several parameters $\eps$ we could adjust independently: the 
target distortion of Talwar's algorithm, the level in the split-tree 
decomposition at which a child town is attached, and the amount of adjustment 
permitted in defining $X_T$. The latter two parameters we set to $\eps$, but the 
distortion in \autoref{thm:Talwar-light} needs to be smaller. We use $\eps'$ for 
the target distortion of this embedding and set $\eps' = \eps^2$.

\section{The expected distortion of the embedding}
\label{sec:stretch}

We now show that the expected distortion of the constructed embedding $H$
is $1+O(\eps)$. For this, we focus on a pair of vertices $u,v \in V$
and argue that
\[ \E[\dist_H(u,v)] \leq (1+O(\eps)) \dist_G(u,v). \]
The high-level idea is rather intuitive: suppose that $\dist_G(u,v) \in 
(r_i,r_{i+1}]$ for some $i$ and let $T \in \mc{T}$ be a town (a)~that contains 
both $u$ and $v$, and (b)~whose child towns separate $u$ and $v$; i.e., $u$ and 
$v$  are in different child towns of $T$. We first argue that there is a {\em 
level-$i$ core hub} $h$ of $T$ that lies on the unique shortest $u$--$v$~path. 

\begin{lem}\label{lem:core-hub}
  Let $u$ and $v$ be vertices that lie in different child towns
  of~$T$, and $i$ be such that $\dist_G(u,v)\in(r_i,r_{i+1}]$. There
  is a core hub $h\in C_i\cap \spc(r_i)$ of $T$ on level $i$ that hits
  the shortest path between $u$ and~$v$.
\end{lem}
\begin{proof}
  By definition, $\spc(r_i)$ must contain some hub $h$ on the shortest
  $u$--$v$ path.  Recall that the town $T$ can be partitioned into its
  core $C_i$ on level $i$ and the child towns on levels at
  least~$i$. If hub $h$ is not a core hub, $h\notin\spc(r_i)\cap C_i$,
  then it is either outside of $T$ or in a child town of $T$ on a
  level at least $i$.

  If $h$ lies in a child town $T'$ of $T$, we can assume w.l.o.g.\
  that $v\notin T'$ since $v$ and $u$ lie in different child towns. As
  a hub on level $i$, $h$ cannot be in a town on level $i$ by
  \autoref{dfn:towns}, so $T'$ is a town on level $i+1$ or above. By
  \autoref{lem:townproperties} we then know that
  $\dist_G(v,h)>r_{i+1}$, but at the same time,
  $\dist_G(v,h)\leq\dist_G(v,u)\leq r_{i+1}$---a contradiction. If $h$
  lies outside of $T$, then by \autoref{lem:townproperties}
  $\dist_G(v,u)\geq\dist_G(v,h)\geq \dist(T,V\setminus T)>r_j$,
  where~$j$ is the level of $T$. However by the same lemma,
  $\dist_G(v,u) \leq \diam(T) \leq r_j$---again a contradiction.
\end{proof}

By \autoref{thm:doubling_dim} it now follows that there is an
approximate core hub $h_X \in X_T$ such that 
\begin{equation}\label{apxch}
 \dist_G(h, h_X) \leq \eps r_i = O(\eps) \dist_G(u,v), 
\end{equation}
since $r_{i+1}/r_i = O(1)$ using our assumption that $c=O(1)$. 
We are also able to show that the expected distances between $u$ and $h_X$ and 
$v$ and $h_X$, respectively, are well preserved by $H$.

\begin{lem}\label{lem:single-stretch}
Let $v$ be a vertex in a child town $T'$ of $T \in \mc{T}$, and let $h_X$ be an 
approximate core hub in~$X_T$. If the distance to the closest sibling town of
$T'$ is $r$, then $\E[\dist_H(v,h_X)]\leq(1+O(\eps))\dist_G(v,h_X)+O(\eps r)$.
\end{lem}

Since $u$ lies in a different child town than $v$ and 
$\dist_G(u,v)\in(r_i,r_{i+1}]$, we get $O(\eps r)=O(\eps\cdot\dist_G(u,v))$ in 
\autoref{lem:single-stretch}. Hence, using triangle inequality, the bound on the 
expected distance in this lemma immediately implies the following:
\begin{align*}
\E[\dist_H(v,u)]\leq&\E[\dist_H(v,h_X)]+\E[\dist_H(h_X,u)]\\
\leq & (1+O(\eps))\dist_G(v,h_X) +(1+O(\eps))\dist_G(h_X,u) 
+O(\eps\cdot\dist_G(u,v))\\
\leq & (1+O(\eps))(\dist_G(v,h)+\dist_G(h,h_X)+\dist_G(h_X,h)+\dist_G(h,u))
+O(\eps)\dist_G(u,v)\\
\leq & (1+O(\eps))\dist_G(v,u),
\end{align*}
where the last equality uses the fact that $h_X$ lies close to a
shortest $u,v$-path (see \eqref{apxch}). Together with the fact that 
$\dist_G(u,v) \leq \dist_H(u,v)$, this implies our stretch bound.

\begin{thm}
The expected stretch of the embedding $H$ of $G$ is $1+O(\eps)$.
\end{thm}

The remainder of this section is devoted to providing a proof of 
\autoref{lem:single-stretch}, for which we will need some further details from 
Talwar's embedding of low doubling metrics into bounded treewidth graphs.

\subsection{The distortion of an embedding for approximate core hubs}

Before proceeding with the proof of \autoref{lem:single-stretch} we
will first need to have a closer look at the properties of Talwar's
split-tree decomposition. We will use these properties to prove that our 
computed embedding $H_X$ of the approximate core hubs $X_T$ has distortion 
$1+O(\eps)$.

\begin{lem}[\cite{talwar2004bypassing}]
\label{lem:splittree}
The split-tree decomposition for a metric $(X,\dist)$ with doubling 
dimension~$d$ and aspect ratio $\alpha$ satisfies the following properties: 
 \begin{enumerate}[(1)]
\item\label{item:levels} there are $\log_2 \alpha + 2$ levels, 
\item\label{item:partition} the clusters on each level $\bar i$ partition $X$,
\item\label{item:diameter} the diameter of a cluster at level~$\bar{i}$ is at 
most $2^{\bar{i}+1}$, and 
\item\label{item:cut-prob} the probability that any points $x,y \in X$ are in 
distinct level $\bar{i}$ clusters is $O(d\cdot \dist(x,y)/2^{\bar{i}})$.
 \end{enumerate}
\end{lem}

Recall the notion of $\delta$-net from \autoref{dfn:netcover}. The 
main result of \citet{talwar2004bypassing} that we use for our embedding 
is the following more detailed account of \autoref{thm:Talwar-light}.

\begin{thm}[\cite{talwar2004bypassing}]
\label{thm:Talwar}
Let $(X,\dist)$ be a metric with doubling dimension $d$ and aspect ratio 
$\alpha$. In polynomial time we can compute a probabilistic embedding $\mc{D}$ 
of $X$ into 
bounded treewidth graphs. In particular, a computed graph $H\in\mc{D}$ has a 
tree decomposition $D$ with the following properties:
 \begin{enumerate}[(i)]
\item each bag $b$ in $D$ corresponds to a cluster $C$ in the split-tree 
decomposition of $(X,\dist)$, and the tree underlying $D$ is precisely
that of the split-tree decomposition;
\item\label{item:net-hierarchy}  the nets of the clusters form a hierarchy, 
i.e., every vertex in a bag $b$ is also contained in one of the children of $b$ 
in the tree $D$;
\item a bag~$b$ corresponding to a cluster $C$ at level $\bar{i}$
  consists of a 
$\beta 2^{\bar{i}}$-net of $C$ for some~$\beta > 0$; and
\item using a $\beta 2^{\bar{i}}$-net for clusters at level $\bar{i}$, the 
expected distortion of $H$ is $1+ O(\beta d \log\alpha)$, and the treewidth of 
$D$ is at most~$(1/\beta)^{O(d)}$.
 \end{enumerate}
In particular there is a $\beta = \Theta(\eps'/(d\log\alpha))$ such that the 
expected distortion is $1+\eps'$, and the treewidth is 
$(d\log(\alpha)/\eps')^{O(d)}$.
\end{thm}

For every bag $b$ in $D$, the graph $H$ contains a complete graph on the nodes 
in $b$. The $\beta 2^{\bar{i}}$-net in each bag serves as a set of 
\emph{portals}, through which connections leaving the cluster are routed, 
analogous to those in~\cite{arora2003survey}. The bound on the stretch follows 
from \autoref{lem:splittree} (see~\cite{talwar2004bypassing} for the details). 
The bound on the treewidth follows from the fact that a $\beta 2^{\bar i}$-net 
in a cluster with diameter at most~$2^{\bar i+1}$ has aspect ratio~$O(1/\beta)$ 
and the following property of low doubling dimension metrics.

\begin{lem}[\cite{gupta2003bounded}]\label{lem:dd_bound}
Let $(X,\dist)$ be a metric with doubling dimension $d$ and $Y\subseteq X$ be a 
set with aspect ratio~$\alpha$. Then $|Y|\leq 2^{d\lceil\log_2\alpha\rceil}$.
\end{lem}

To analyze the distortion of the embedding $H_X$, we rely on the following 
useful fact that relates properties of hubs in $X_T$ and their 
representatives in~$Y_T$. 
Recall that a cluster $C_X$ of $X_T$ is formed from a cluster $C_Y$ 
of $Y_T$ by expanding each hub $h\in C_Y$ into all vertices in $X_T$ that $h$ 
represents, and a bag $b_X$ of the tree decomposition $D_X$ of $X_T$ is formed 
by the same procedure from a bag $b_Y$ of the tree decomposition $D_Y$ of~$Y_T$. 
For such pairs of clusters and bags we obtain the following.

\begin{lem}\label{lem:netstretch}
If $b_Y$ is a $\delta$-net of $C_Y$ for some $\delta$, then $b_X$ is a 
$2\delta$-cover of $C_X$, i.e., for each $h_X \in C_X$ there is a $h_Y \in b_X$ 
such that $\dist_G(h_X,h_Y) \le 2\delta$.
\end{lem}
\begin{proof}
Let $h_X\in C_X$. If $h_X \in b_X$, we are done. If not, let $h_Y$ be $h_X$'s 
representative in $Y_T$, and let $T'$ be the child town of $T$ for which 
$h_X,h_Y\in X_T\cap T'$.
We obtained~$b_X$ by expanding each $h\in b_Y$ 
into all vertices it represents, so $h_X\notin b_X$ implies
$h_Y\notin b_Y$. Let $h'_Y \in b_Y$ be the closest vertex in $b_Y$ to $h_Y$.
The set $b_Y$ is a $\delta$-net of $C_Y$, so $\dist_G(h_Y,h'_Y) \le \delta$, but
$h'_Y \notin T'$, since $h'_Y\neq h_Y$, and each 
town contains at most one representative. By 
\autoref{lem:townproperties}, $\diam(T') \le \dist_G(T',V\setminus T')$, so
$\dist_G(h_X,h_Y) \le \dist(h_Y,h'_Y)$, which 
means that $\dist_G(h_X,h'_Y) \le 2\delta$. Finally, $h'_Y \in b_X$, since $b_Y 
\subseteq b_X$.
\end{proof}

Another useful tool is given by the following lemma, which compares the 
separation probabilities of approximate core hubs and their representatives.

\begin{lem}\label{lem:sep-prob}
Let $u,v\in X_T$ be two hubs with respective representatives $u',v'\in Y_T$. If 
$u'\neq v'$, then the probability with which $u$ and $v$ are in distinct level 
$\bar i$ clusters is $O(d\cdot\dist_G(u,v)/2^{\bar i})$, where $d$ is the 
doubling dimension of $Y_T$.
\end{lem}
\begin{proof}
If the representatives $u'$ and $v'$ of $u$ and $v$ differ, then $u$ and $v$ 
must lie in different child towns $T'$ and $T''$ of $T$. By 
\autoref{lem:townproperties}, $\diam_G(T') < \dist_G(T',V\setminus T')\leq 
\dist_G(T',T'')$, so that $\dist_G(u,u') \le \dist_G(u,v)$, and similarly for 
$\dist_G(v,v')$. Hence $\dist_G(u',v') \le \dist_G(u',u) + \dist_G(u,v) + 
\dist_G(v,v') \le 3\cdot\dist_G(u,v)$. By \autoref{lem:splittree} 
\eqref{item:cut-prob}, the separation probability of $u'$ and $v'$ on level 
$\bar i$ is at most $O(d\cdot\dist_G(u',v')/2^{\bar i})$. Since $u$ and $v$ lie 
in different clusters if and only if their representatives do, the probability 
of $u$ and $v$ being separated is $O(d\cdot\dist_G(u,v)/2^{\bar i})$.
\end{proof}

The next lemma bounds the distortion of $H_X$. Its proof closely mirrors 
Talwar's proof of \autoref{thm:Talwar} (c.f.~\cite{talwar2004bypassing}).

\begin{lem}
\label{lem:xembedding}
If the embedding $H_Y$ of $(Y_T,\dist_G)$ is computed according to 
\autoref{thm:Talwar}, then the constructed embedding $H_X$ of $(X_T,\dist_G)$ 
has expected distortion $1+O(\eps')$.
\end{lem}
\begin{proof}
Consider a cluster $C_Y$ on level $\bar i$ in the split-tree decomposition of 
$Y_T$ given by \autoref{lem:splittree}. For any $h\in C_Y$ the \emph{$\bar 
i$-parent} of $h$ is the closest vertex to $h$ in the bag $b_Y$ corresponding to 
$C_Y$. Since by \autoref{thm:Talwar} the bag $b_Y$ consists of a $\beta 2^{\bar 
i}$-net of $C_Y$, the distance between $h$ and its $\bar i$-parent is at most 
$\beta 2^{\bar i}$. Let $C_X$ be the cluster in $X_T$ formed by expanding each 
$h\in C_Y$ into all vertices in $X_T$ that $h$ represents, and let $b_X$ be the 
corresponding bag formed by the same procedure from $b_Y$. We define the $\bar 
i$-parent of a vertex $w\in C_X$ in the same way as for $C_Y$, i.e.\ it is the 
closest vertex to $w$ in~$b_X$. According to \autoref{lem:netstretch}, the 
distance from $w$ to its $\bar i$-parent is at most $2\beta 2^{\bar i}$.

For an arbitrary pair $u,v\in X_T$ we bound the distortion of their distance in 
$H_X$ by considering the path along the $\bar i$-parents of $u$ and $v$ for 
increasing values of $\bar i$. More concretely, since the bags of the tree 
decomposition $D_Y$ of $H_Y$ form a hierarchy by \autoref{thm:Talwar}, the same 
is true for the bags of the tree decomposition $D_X$ of $H_X$. Thus on the 
lowest level $\bar l$ of the split-tree decomposition, the $\bar l$-parent of a 
vertex $w$ is $w$ itself. We inductively define $v_{\bar l}=v$, $u_{\bar l}=u$, 
and $v_{\bar i}$ and $u_{\bar i}$ to be the $\bar i$-parent of $v_{\bar i-1}$ 
and $u_{\bar i-1}$, respectively, for any level $\bar i>\bar l$. Since the bags 
of $D_X$ form a hierarchy, for each level $\bar i>\bar l$ the edges $\{u_{\bar 
i-1},u_{\bar i}\}$ and $\{v_{\bar i-1},v_{\bar i}\}$ exist in $H_X$. Thus the 
distance from $u_{\bar i-1}$ to $u_{\bar i}$ and from $v_{\bar i-1}$ to $v_{\bar 
i}$ is at most $2\beta 2^{\bar i}$ in $H_X$. Now, let $\bar j$ be the lowest 
level at which $u$ and $v$ lie in the same cluster of $X_T$. In particular, the 
$\bar j$-parents $u_{\bar j}$ and $v_{\bar j}$ lie in the same bag of $D_X$, and 
so there is an edge $\{u_{\bar j},v_{\bar j}\}$ in $H_X$. We next bound the 
expected length of the path $P=(u=u_{\bar l},u_{\bar l+1},\ldots,u_{\bar 
j},v_{\bar j},v_{\bar j-1},\ldots,v_{\bar l}=v)$ in $H_X$ in terms of 
$\dist_G(u,v)$.

For this we need to bound the probability with which any pair of $\bar 
i$-parents $u_{\bar i}$ and $v_{\bar i}$ lie in different clusters of $X_T$ on 
level $\bar i$. Note that $u$ and $v$ always lie in the same cluster as their 
respective $\bar i$-parent, and so $u_{\bar i}$ and $v_{\bar i}$ lie in 
different clusters of $X_T$ on level $\bar i$ if and only if $u$ and $v$ lie in 
different clusters of $X_T$ on the same level. \autoref{lem:splittree} gives a 
bound for the probability with which representatives lie in different clusters 
of $Y_T$ in terms of the distance between them. Let $u',v'\in Y_T$ be the 
respective representatives of $u$ and $v$. If $u'=v'$ then obviously 
$\dist_G(u',v')=0$. Otherwise, $u'$ and $v'$ lie in different child towns of 
$T$. By \autoref{lem:sep-prob}, this means that $u$ and $v$ lie in different 
clusters on level $\bar i$ with probability $O(d\cdot\dist_G(u,v)/2^{\bar i})$. 
Let $A_{\bar i}$ be the indicator variable that is $1$ if $u_{\bar i}$ and 
$v_{\bar i}$ lie in different clusters of $X_T$ on level $\bar i$, and $0$ 
otherwise, so that $\Pr[A_{\bar i}=1]=O(d\cdot\dist_G(u,v)/2^{\bar i})$.

Consider the subpaths of $P$ from $u$ to $u_{\bar j}$ and $v$ to $v_{\bar j}$. 
The length of each such path is at most $\sum_{\bar i} 2\beta 2^{\bar i+1} 
A_{\bar i}$. Accordingly, the edge $\{u_{\bar j},v_{\bar j}\}$ has length at 
most $\dist_G(u,v)+2\sum_{\bar i} 2\beta 2^{\bar i+1} A_{\bar i}$. Since there 
are at most $\log_2\alpha$ levels in the split-tree decomposition, we can bound 
the expected length of $P$ by
\[
\dist_G(u,v)+4\sum_{\bar i=\bar l}^{\log_2\alpha} 2\beta 2^{\bar i+1}\cdot 
O(d\cdot\dist_G(u,v)/2^{\bar i})=\\
(1+O(\beta d\log\alpha))\dist_G(u,v)=
(1+O(\eps'))\dist_G(u,v),
\]
where we use that $\beta = \Theta(\eps'/(d\log \alpha))$ by 
\autoref{thm:Talwar}.
\end{proof}

\subsection{The distortion of the embedding of the graph}

We now turn to proving \autoref{lem:single-stretch}. For this, throughout this 
section, we focus on a town $T$ of the towns decomposition $\mc{T}$. We further 
let $T'$ be some child town of $T$, and we let the distance $r$ between $T'$ and 
the closest sibling town be in the interval $(r_i,r_{i+1}]$. Further, we define 
$b$ to be the connecting bag of $T'$ (c.f.~\autoref{alg_embedding}), and let $C$ 
be the corresponding cluster in the split-tree decomposition of the approximate 
core hubs $X_T$.

Given vertex $v \in T' \subseteq T$, and some core hub $h_X \in X_T$, the goal 
is to bound their expected distance in the constructed embedding $H$ in terms of 
their distance in the input graph $G$. If $H$ contains an edge between $v$ and 
$h_X$ then we are of course immediately done, but this may not be the case. For 
example, in the construction of the embedding, we add direct links between 
vertices of $T'$ and members (i.e., {\em net points}) of the connecting bag $b$, 
but $h_X$ may not be a member of $b$. We first consider this issue and show 
that, even if $h_X \in C\setminus b$, then $b$ at least contains a net point 
close to $h_X$.

\begin{lem}\label{lem:nearby}
  For any approximate core hub $h\in X_T \cap C$, the bag $b$ contains a net
  point $w$ such that $\dist_H(h,w)=O(\eps r_i)$.
\end{lem}
\begin{proof}
Let $\bar l$ be the level of $b$ in the tree decomposition~$D_X$, which by 
\autoref{alg_embedding} is at most $\bar i+\log_{2}(d/\eps)$, where $\bar 
i=\lceil\log_2 r_i\rceil$. 
If $h\in b$ there is nothing to show.
By \eqref{item:net-hierarchy} of \autoref{thm:Talwar}, the bags of $D_Y$ form a 
hierarchy, which by construction of $D_X$ means that the bags of $D_X$ do too. 
Thus $h\notin b$ is a vertex in a bag on some level below $\bar l$, and so we 
can reach some vertex of $b$ from $h$ in $H_X$ by starting at the bag containing 
$h$ and following the edges to higher level bags until we reach $b$. More 
concretely, the bags computed for the tree decomposition $D_Y$ of the 
representative hubs~$Y_T$ contain $\beta 2^{\bar j}$-nets of the corresponding 
clusters by \autoref{thm:Talwar}. Hence by \autoref{lem:netstretch}, the bags of 
$D_X$ contain $2\beta 2^{\bar j}$-covers of the clusters of $X_T$. Thus there is 
a path in $H_X$ from $h$ to some vertex $w$ of the bag $b$ that traverses the 
net points of the bags up the levels until reaching $\bar l$, by always moving 
to the closest net-point at the next level. The length of this path is at most 
$\sum_{\bar j=1}^{\bar l}2\beta 2^{\bar j}=O(\beta 2^{\bar l})=O(\beta 
dr_i/\eps)$, since $2^{\bar l}=O(dr_i/\eps)$. Because 
$\beta=O(\eps'/(d\log\alpha))$ by \autoref{thm:Talwar} and $\eps'=\eps^2$, this 
bound simplifies to~$O(\eps r_i)$, which also bounds $\dist_H(h,w)$.
\end{proof}

\begin{wrapfigure}[10]{R}{0.43\textwidth}
\centering{\includegraphics[width=0.42\textwidth]{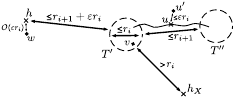}}
\caption{\label{fig:net-point} The net point $w$ lies in the connecting bag $b$ 
of $T'$, and $h$ lies in the corresponding cluster $C$. Note that $v$ may be 
closer to $h_X$ than to $h$.}
\end{wrapfigure}

The above lemma provides a vertex $w$ of the connecting bag $b$ of $T'$ through 
which we can connect to a hub $h_X$, if $h_X\in C$. In case $h_X$ lies outside 
of $C$ however, as we will see the following lemma provides such a vertex in $b$ 
to connect to $h_X$.

\begin{lem}\label{lem:detour}
For any $v\in T'$ and approximate core hub $h_X\in X_T\setminus T'$, the 
connecting bag $b$ of $T'$ contains a vertex $w$ such that 
$\dist_G(v,w)=O(\dist_G(v,h_X))$ and $\dist_G(v,w)=O(r_{i})$.
\end{lem}
\begin{proof}
Recall that, by our choice in \autoref{alg_embedding}, cluster $C$ corresponding 
to connecting bag $b$ of $T'$ contains the closest hub $h\in X_T$ to~$T'$. By 
\autoref{lem:nearby}, there exists $w\in b$ with 
$\dist_G(h,w)\leq\dist_H(h,w)=O(\eps r_i)$ (cf.\ \autoref{fig:net-point}). As 
by triangle inequality $\dist_G(v,w)\leq\dist_G(v,h)+\dist_G(h,w)$, it remains 
to show that $\dist_G(v,h)=O(r_i)$ in order to prove $\dist_G(v,w)=O(r_i)$, if 
$\eps$ tends to zero. By \autoref{lem:core-hub} there is a core hub $u$ of $T$ 
on level $i$, which lies on the shortest path between $T'$ and $T''$, the closest
sibling town to $T'$, and thus $u$ is at most as far from $T'$ as any vertex in $T''$.
Hence $\dist_G(T',u)\leq r_{i+1}$, since we assumed that the distance~$r$ 
between $T'$ and $T''$ lies in the interval~$(r_i,r_{i+1}]$. By 
\autoref{thm:doubling_dim} there is an approximate core hub $u'\in X_T$ for 
which $\dist_G(u,u')\leq\eps r_i$.  Hence the closest approximate core hub~$h$ 
is at distance at most $r_{i+1}+\eps r_i$ from~$T'$. From 
\autoref{lem:townproperties} it follows that every town on level at least $i+1$ 
has distance more than $r_{i+1}$ to any other town. Since the distance $r$ from 
$T'$ to $T''$ is at most $r_{i+1}$, the level of $T'$ is at most $i$.  Hence the 
same lemma also implies that the diameter of $T'$ is at most~$r_i$, and thus 
$\dist_G(v,h)\leq\diam(T')+\dist_G(T',h)\leq r_{i+1}+(1+\eps)r_i=O(r_i)$, since 
$r_{i+1}/r_i = c$ is constant and we assume that $\eps$ tends to zero. This 
implies $\dist_G(v,w)=O(r_{i})$ as claimed. Note that since $h_X$ lies outside 
of~$T'$, $\dist_G(v,h_X)\ge\dist_G(T',V\setminus T')>r_i$ by 
\autoref{lem:townproperties}, which immediately implies the remaining bound 
$\dist_G(v,w)=O(\dist_G(v,h_X))$.
\end{proof}

So far we have identified vertices $w$ in the connecting bag $b$ through which 
we are able to connect to a hub $h_X$ from a vertex $v\in T'$ for the two cases 
when $h_X\in C$ and $h_X\notin C$. The next lemma provides a bound on the 
probability with which we need to consider each of these cases. Additionally it 
also bounds the distance from $v$ to $h_X$ in the former case.

\begin{lem}\label{lem:direct}
Let $h_X$ be an approximate core hub in $X_T$, and $v \in T'$, then
$\Pr[h_X\notin C] = O(\eps\cdot\dist_G(v,h_X)/r_i)$. In addition, 
$\dist_H(v,h_X)\leq \dist_G(v,h_X)+O(\eps r_i)$ if $h_X\in C$.
\end{lem}
\begin{proof}
If $h_X\in T'$ then $h_X\in C$, since by \autoref{alg_embedding} the cluster $C$ 
contains the closest approximate core hub to $T'$ and all hubs of $X_T$ that are 
represented by the same hub of $Y_T\cap C$ (i.e.\ that are of the same child 
town) are contained in $C$. Hence if $h_X\notin C$ then $h_X\notin T'$. Consider 
the vertex $w\in b$ for which $\dist_G(v,w)=O(\dist_G(v,h_X))$, which now exists 
due to \autoref{lem:detour}. The hub $h_X$ is in $C$ if and only if $w$ and 
$h_X$ are in the same cluster on the level~$\bar l$ of~$C$. If the level $\bar 
l$ of the cluster $C$ is the level $\bar j$ of the root of $D_X$, then $C$ 
contains all vertices of $T$ including~$h_X$ and $w$, and so if $h_X\notin C$ 
then $\bar l\neq\bar j$. If $w$ and $h_X$ have the same representative, they 
will be in the same cluster by \autoref{alg_embedding}, so that if $h_X\notin C$ 
then $w$ and $h_X$ have different representatives in $Y_T$. 

By these observations, the probability with which $w$ and $h_X$ lie in different 
clusters is $O(d\cdot\dist_G(w,h_X)/2^{\bar l})$ using \autoref{lem:sep-prob}, 
which in turn can be bounded by $O(\eps\cdot\dist_G(w,h_X)/r_i)$, as $2^{\bar 
l}=\Theta(dr_i/\eps)$ by \autoref{alg_embedding} whenever $\bar l\neq\bar j$. 
Upper bounding $\dist_G(w,h_X)$ in terms of 
$\dist_G(w,v)+\dist_G(v,h_X)=O(\dist_G(v,h_X))$ we obtain $\Pr[h_X\notin 
C]=O(\eps\cdot\dist_G(v,h_X)/r_i)$.

To bound the distance if $h_X\in C$, by \autoref{lem:nearby} we know that there 
is a vertex $h_b\in b$ such that $\dist_H(h_b,h_X)=O(\eps r_i)$, and $v$ has an 
edge in $H$ to $h_b$. Therefore 
$\dist_H(v,h_X)\leq\dist_H(v,h_b)+\dist_H(h_b,h_X)=\dist_G(v,h_b)+O(\eps r_i)$. 
Since $\dist_G(h_b,h_X)\leq\dist_H(h_b,h_X)$, we can upper bound 
$\dist_G(v,h_b)$ by $\dist_G(v,h_X)+\dist_H(h_X,h_b)$, which proves the claim.
\end{proof}

\autoref{lem:direct} provides a bound on the distance between vertices
of $T'$ and approximate core hubs in~$C$.  We also need to bound the
distance between vertices of $T'$ and core hubs of $T$ that are not in~$C$. The 
following lemma will be useful in this endeavour.

\begin{lem}
\label{lem:conditionalexpectation}
Let $H_X$ be the probabilistic embedding of $(X_T,\dist_G)$ with expected 
distortion $1+O(\eps')$ given by \autoref{lem:xembedding}. Let $x,y \in X_T$, 
and let $C$ be a cluster in the randomized split-tree decomposition containing 
$x$. Then
$\E[\dist_{H_X}(x,y) \mid y \notin C] \le (1 + O(\eps')/\Pr[y \notin 
C])\dist_G(x,y).$
\end{lem}
\begin{proof}
  By \autoref{lem:xembedding},
  the expected distance between $x$ and $y$ in $H$ is at most
  $(1+O(\eps'))$ times their distance in metric $(X_T,\dist_G)$, and
  hence 
  \begin{align*}
    \E[\dist_{H_X}(x,y)] & =  \Pr[y \notin C]\E[\dist_{H_X}(x,y) \mid y \notin 
C] + \Pr[y \in C] \E[\dist_{H_X}(x,y) \mid y \in C]\\
                     & \leq (1 + O(\eps'))\dist_G(x,y).
  \end{align*}
  Embedding $H_X$ dominates $(X_T,\dist_G)$, and hence 
  $\E[\dist_{H_X}(x,y) \mid y \in C] \geq \dist_G(x,y)$. The inequality above
  therefore implies that 
  \[ \Pr[y \not\in C] \E[\dist_{H_X}(x,y) \mid y \not\in C] + (1-\Pr[y
  \not\in C]) \dist_G(x,y) \leq (1+ O(\eps'))\dist_G(x,y). \]
  Rearranging,
  $\Pr[y \notin C](\E[\dist_{H_X}(x,y) \mid y \notin C] - \dist_G(x,y)) \le
  O(\eps')\dist_G(x,y)$, and
\[
\E[\dist_{H_X}(x,y) \mid y \notin C] \le \left(1 + \frac{O(\eps')}{\Pr[y \notin
    C]}\right) \dist_G(x,y) \; .
\]
\end{proof}

We are now ready to bound the distance between a vertex $v \in T$ and any core 
hub in $X_T$, given the tools of the above lemmas.

\begin{proof}[Proof of \autoref{lem:single-stretch}.]
Let $C$ be the cluster corresponding to the connecting bag $b$ of $T'$. We 
bound $\E[\dist_H(v,h_X)]$ in terms of the conditional expected 
values $\E[\dist_H(v,h_X) \mid h_X\in C]$ and $\E[\dist_H(v,h_X) \mid h_X\notin 
C]$. 
If $h_X\in C$ we get a (deterministic) bound on the distance between $v$ and 
$h_X$ from \autoref{lem:direct}. Hence $\E[\dist_H(v,h_X) \mid h_X\in C]\leq 
\dist_G(v,h_X)+O(\eps r_i)$.

If $h_X \in T'$ then $h_X\in C$, since $C$ contains the closest hub to $T'$ and 
all hubs of $X_T$ in the same child town of $T$ end up in the same cluster after 
expanding all hubs of $Y_T$ into the hubs of $X_T$ that they represent. Hence 
if $h_X\notin C$ then $h_X\notin T'$, and by \autoref{lem:detour} there is a 
vertex $w\in b$ for which $\dist_G(v,w)=O(\dist_G(v,h_X))$. Both $w$ and $h_X$ 
are approximate core hubs, and so $\E[\dist_H(w,h_X)\mid 
h_X\notin C]\leq\E[\dist_{H_X}(w,h_X)\mid h_X\notin C]$, as $H$ contains $H_X$. 
Applying \autoref{lem:conditionalexpectation} on this conditional expected 
distance, we obtain
\begin{align*}
\Pr[h_X \notin C]\E[\dist_H(v,h_X) \mid h_X&\notin C] \le 
\Pr[h_X \notin C](\dist_G(v,w) + \E[\dist_H(w,h_X) \mid h_X \notin C]) \\
& \le \Pr[h_X \notin C]\left(\dist_G(v,w) + 
\left(1 + \frac{O(\eps')}{\Pr[h_X \notin C]}\right)\dist_G(w,h_X)\right)
  \\
& = 
\Pr[h_X \notin C](\dist_G(v,w) + \dist_G(w,h_X)) + O(\eps') \dist_G(w,h_X) \\
& \le
\Pr[h_X \notin C](2\cdot\dist_G(v,w) + \dist_G(v,h_X))\\
& \quad + 
O(\eps') (\dist_G(w,v) + \dist_G(v,h_X))\\
& =
\Pr[h_X \notin C](2\cdot\dist_G(v,w) + \dist_G(v,h_X)) + 
O(\eps'\cdot\dist_G(v,h_X)).
\end{align*}

From \autoref{lem:detour} we also know that $\dist_G(v,w)=O(r_i)$. Additionally
using that $\eps' = \eps^2$, and the bound on $\Pr[h_X \notin C]$ in 
\autoref{lem:direct}, the expression above is
\begin{multline*}
\Pr[h_X \notin C]\dist_G(v,h_X) + 
O\left(\eps\frac{\dist_G(v,h_X)}{r_i}\right)O(r_{i})
+ O(\eps'\cdot\dist_G(v,h_X)) = \\
\Pr[h_X \notin C]\dist_G(v,h_X) + O(\eps) \dist_G(v,h_X) \; .
\end{multline*}
Combining the above bounds we obtain
\begin{align*}
\E[\dist_H(v,h_X)] =&\Pr[h_X\in C]\E[\dist_H(v,h_X) \mid h_X\in C] 
+\Pr[h_X\notin C]\E[\dist_H(v,h_X) \mid h_X\notin C]\\
\leq& \Pr[h_X\in C](\dist_G(v,h_X) + O(\eps r_i))
+ \Pr[h_X\notin C] \dist_G(v,h_X) \\
&+ O(\eps)\dist_G(v,h_X)\\
=& (1+O(\eps))\dist_G(v,h_X)+O(\eps r_i) \; ,
\end{align*}
where $r_i=\Theta(r)$, which proves the claim.
\end{proof}

\section{The doubling dimension of approximate core hubs}
\label{sec:dd}

The aim of this section is to give a proof of \autoref{thm:doubling_dim} by 
showing that for any town $T\in\mc{T}$ there is a set $X_T\subseteq T$ of 
approximate core hubs with bounded doubling dimension. We first define
 the set $X_T$ and then compare its properties with those of the core hubs. 
In particular, even though we obtain the approximate core hubs by shifting the 
core hubs to positions nearby, the resulting set is still locally sparse on 
each level. In addition, they are also \emph{locally nested}. Roughly 
speaking, this means that within a small ball of radius $\eps r_i$ 
for some level $i$, all approximate core hubs above level $i$ are 
``nested'', i.e., contained in one another. This property will help us in 
bounding the doubling dimension of $X_T$ independently of the aspect 
ratio.

The set $X_T$ of a town $T$ of level $j$ is the union of sets $X_T^i$, one for 
each level $i\in\{1,\ldots,j-1\}$, which are defined inductively as follows in 
\autoref{alg:xt}. We call a vertex in $X_T^i$ an \emph{approximate core hub of 
$T$ on level~$i$}. Recall that $C_i$ is the core of $T$ at level $i$ 
(\autoref{dfn:core}), and $C_0 = \emptyset$ since the sprawl is empty on level 
$0$.

\begin{algorithm2e}[H]
\caption{Defining $X_T$}
\label{alg:xt}
\DontPrintSemicolon
$X_T^1 \leftarrow C_1\cap\spc(r_1)$ \;
\For{$i = 2,\ldots, j-1$} {

$X_T^i \leftarrow \emptyset$ \;
\ForEach{$h\in C_i\cap \spc(r_i)$} {
\lIf{$\exists h' \in X_T^l$ for some $l < i$ such that $\dist(h,h') \le \eps r_i$} {
add $h'$ to $X_T^i$
}\lElse{
add $h$ to $X_T^i$
}
}
}
\Return $\bigcup_{i=1}^{j-1} X_T^i$
\end{algorithm2e}

Note that this definition of $X_T$ fulfills the two properties of 
\autoref{thm:doubling_dim} that there must be an approximate core hub $h'\in 
X_T$ within distance $\eps r_i$ of each core hub $h$ of level $i$ and 
that $X_T$ can be computed in polynomial time. Note also that $X_T^i\subseteq 
\bigcup_{l=1}^i C_l\cap\spc(r_l)$, and hence the vertices in $X_T^i$ are 
core hubs, but not necessarily core hubs of level $i$. 
The main benefit of shifting core hubs to approximate core hubs 
is that for any town $T\in\mc{T}$ on level 
$j$, the set system $\{X_T^i\}_{i=1}^j$ is \emph{locally nested} as
we explain in the following lemma. 

\begin{lem}\label{lem:nested}
Let $B$ be a set of diameter at most $\eps r_l$ for some level $l$, and let 
$i$ be the lowest level for which $X_T^i\cap B\neq\emptyset$. 
The approximate core hubs on level $q \geq \max\{i,l\}$ in $B$ must also
be core hubs on some level at most $\max\{l,i\}$; i.e.,
$B \cap X_T^q \subseteq \bigcup_{p=1}^{\max\{l,i\}} X_T^p$.
\end{lem}
\begin{proof}
The statement is trivially true for $q=\max\{l,i\}$. Consider any higher 
level~$q>\max\{l,i\}$. 
Since the diameter of $B$ is at most $\eps r_l\leq \eps 
r_q$ and $X_T^i\cap B\neq\emptyset$, for every $h\in B \cap C_q \cap \spc(r_q)$ 
there is a vertex $h'\in X_T^i$ at distance at most $\eps r_q$ from $h$. 
Hence by the definition of the approximate core hubs in \autoref{alg:xt}, 
$X_T^q\cap B \subseteq \bigcup_{p=1}^{q-1} X_T^p$, and the claim follows
by induction.
\end{proof}

The cost of using approximate core hubs is that it is not immediately clear why 
the vertices in $X_T^i$ should still be locally sparse. This requires a tricky 
argument that we turn to now. The crucial observation leading to this result is 
that we can bound the number of hubs of a shortest path cover $\spc(r_i)$ not 
only in a ball $B_{cr_i/2}(v)$ using the local sparsity but also close to the 
ball. The approximate core hubs in $X_T^i$ are obtained by shifting the core 
hubs of level~$i$ to lower level core hubs at distance at most $\eps r_i$. Hence 
the number of vertices of $X_T^i \cap B_{cr_i/2}(v)$ can be bounded by the total 
number of level $i$ core hubs that are within distance $\eps r_i$ of 
$B_{cr_i/2}(v)$. The definition of highway dimension (\autoref{dfn:hd}) allows 
us to get a handle on the hubs in larger balls of radius~$cr_i$, and this, 
combined with the minimality of our shortest path cover, allows us to bound the 
number of nearby core hubs. Specifically, in a graph of highway dimension $k$, 
and given a locally $s$-sparse shortest path cover, we are able to show that the 
approximate core hubs $X_T^i$ of level $i$ are locally $3ks$-sparse as long as 
the stretch parameter $\eps$ is chosen to be at most $2$. The lemma is stated in 
a slightly more general form than we need it here, since we will reuse it later.

\begin{lem}\label{lem:hub_bound}
  For a metric $(V,\dist_G)$ induced by an underlying graph $G$ of highway
  dimension~$k$, let $B_{cr/2}(v)$ be a ball of radius $cr/2$ centered
  at $v\in V$, and let $\spc(r)$ be a minimal locally $s$-sparse
  shortest path cover. There are at most $3sk$ hubs $h \in \spc(r)$
  for which $\dist_G(h,B_{cr/2}(v)) \le cr/2$.
\end{lem}

We note that this lemma does not bound the number of hubs in $\spc(r)$ that lie 
in a ball $B_{cr}(v)$, and in fact the number of hubs in $B_{cr}(v)\cap\spc(r)$ 
can be unbounded: in a star with edges of length $cr$ a minimal shortest path 
cover $\spc(r)$ may contain all vertices except the center vertex $v$ of the 
star. This shortest path cover is also locally $1$-sparse, since any ball of 
radius $cr/2$ contains only one vertex of the star. However the ball of radius 
$cr$ centered at $v$ contains the whole star, and thus all hubs from $\spc(r)$, 
i.e.\ a potentially unbounded number.

Since the hubs considered in \autoref{lem:hub_bound} may lie outside of 
$B_{cr/2}(v)$, we need to use \autoref{dfn:hd}, which bounds the number of hubs 
in larger balls of radius~$cr$.  However, the hubs given by \autoref{dfn:hd} do 
not necessarily coincide with those of $\spc(r)$. Therefore, we need an 
additional tool, as given by the following technical lemma, which relates the 
hubs given by \autoref{dfn:hd} with those in $\spc(r)$.

In the following lemma, we consider once more a metric induced by
graph $G=(V,E)$ of highway dimension $k$. As usual, we let $\spc(r)$
denote a locally $s$-sparse shortest-path cover for radius $r$. 
Consider radii $r, \tilde r$ such that $\tilde r < cr/2$, and
let $B_{c\tilde r}(v)$ be a ball of radius $c\tilde r$ centered at $v$. For each vertex
$h \in B_{c\tilde r}(v) \cap \spc(r)$, we let $P_h$ be a shortest path
that (a) lies in $B_{c\tilde r}(v)$, i.e.\ $V(P_h)\subseteq B_{c\tilde r}(v)$, 
(b)~has length in $(\tilde r, cr/2]$, and (c) contains $h$. If no such path 
exists, we let $P_h= \bot$. 

\begin{lem}\label{lem:balls_aux}
Let $\tilde W$ be the set of all vertices $h\in B_{c\tilde r}(v) \cap
\spc(r)$ for which $P_h\neq \bot$. 
Then $|\tilde W|\leq sk$.
\end{lem}
\begin{proof}
  The proof follows directly from \autoref{dfn:hd}. The definition
  implies that there is a set $K\subseteq B_{c\tilde r}(v)$ of at most
  $k$ vertices covering all shortest paths in $B_{c\tilde r}(v)$ of
  length more than $\tilde r$. In particular these vertices cover each
  path $P_h$ for $h \in \tilde W$.  We have $h \in V(P_h)$ and the
  length of $P_h$ is at most $cr/2$, so $\dist_G(h,K) \le cr/2$.
  Therefore $\tilde W$ can be covered by at most $k$ balls of radius
  $cr/2$ centered at each vertex in $K$.  The set $\spc(r)$, and with
  that also $\tilde W$, is locally $s$-sparse, so each of these balls
  contains at most $s$ nodes, yielding $|\tilde W| \le sk$.
\end{proof}

We now prove \autoref{lem:hub_bound}. For this, define 
\[ W = \{ h \in\spc(r) \mid \dist_G(h,B_{cr/2}(v)) \le cr/2\} \]
as the set of hubs near $v$ whose size we want to bound.  
In order to accomplish this, we carefully choose three radii
$\tilde{r}_i$, where $i\in\{1,2,3\}$, and let $\tilde{W}_i$ be the 
corresponding set of hubs as defined in \autoref{lem:balls_aux} (see 
\autoref{fig:hub_bound}).
We will then show that 
\[ W \subseteq \tilde{W}_1 \cup \tilde{W}_2 \cup \tilde{W}_3, \]
and conclude that $W$ has at most $3sk$ elements directly from
\autoref{lem:balls_aux}.

\begin{figure} %
\centering{\includegraphics[width=\textwidth]{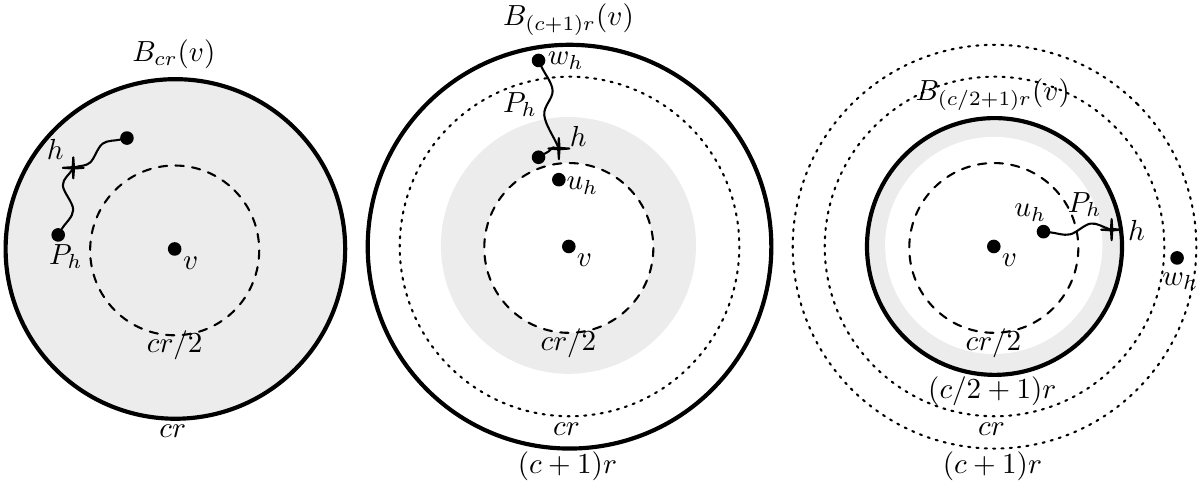}}
\caption{\label{fig:hub_bound} 
The three balls in \autoref{lem:hub_bound}. The dashed ball is $B_{cr/2}(v)$,
and the bold balls are the three considered balls $B_{c\tilde r_i}(v)$, 
moving from left to right. Hubs are crosses, and
shaded areas represent possible locations for hubs.
Hubs in $\tilde W_1$ (left) cover paths entirely in $B_{cr}(v)$. For a hub $h\in 
\tilde W_2$ (center) the path $P_h$ between $h$ and $w_h$ is long, while for 
$h\in \tilde W_3$ (right) the path $P_h$ from $h$ to $u_h$ is long.}
\end{figure}

\begin{proof}[Proof of \autoref{lem:hub_bound}]
We first apply \autoref{lem:balls_aux} for $\tilde{r}_1=r$, and infer that the 
set $\tilde{W}_1$ of hubs $h \in \spc(r)$ that cover a shortest path contained 
in $B_{cr}(v)$ and with length in $(r,cr/2]$, is at most~$sk$.

Observe that, by the inclusion-wise minimality of $\spc(r)$, each $h \in 
\spc(r)$ must hit some shortest path $Q_h$ with length in $(r,cr/2]$. For $h \in 
W\setminus\tilde{W}_1$ this path $Q_h$ is not contained in $B_{cr}(v)$. Let 
$w_h$ be a vertex on path $Q_h$ of maximum distance from $v$, which by 
assumption must lie outside the ball~$B_{cr}(v)$. We know $\dist_G(h,w_h) \le 
cr/2$, as the distance between $h$ and $w_h$ is bounded by the maximum length of 
$Q_h$. Also let $u_h$ be the closest vertex in $B_{cr/2}(v)$ to $h$. By the 
definition of $W$, $\dist_G(u_h,h)\leq cr/2$. Since $h$ does not cover any 
shortest path inside $B_{cr}(v)$ with length in $(r,cr/2]$, we must have 
$\dist_G(u_h,h) \le r$. Combining these, the distance from $v$ to $w_h$ is at 
most
\[
\dist_G(v,u_h)+\dist_G(u_h,h)+\dist_G(h,w_h)\leq cr/2+r+cr/2= (c+1)r=c(1+1/c)r.
\]
Hence, $Q_h$ lies in the ball $B_{c\tilde{r}_2}(v)$ if we choose 
$\tilde{r}_2=(1+1/c)r$. Furthermore, $h \in \tilde{W}_2$ if $Q_h$
has length in the interval $(\tilde{r}_2, cr/2]$. 

Finally, let us consider a hub
$h \in W\setminus (\tilde{W}_1 \cup \tilde{W}_2)$, for which
the length of the path $Q_h$
must lie in the interval
$(r,\tilde{r}_2]=(r,(1+1/c)r]$. Let $u_h$ and $w_h$ be defined as before.  
The distance between $h$ and $w_h$ is now at most $(1+1/c)r$, while the 
distance between $u_h$ and $w_h$ is more than $cr/2$, as $u_h\in B_{cr/2}(v)$ 
and $w_h\notin B_{cr}(v)$. It follows that
\begin{equation}\label{eq:lenPh} \dist_G(u_h,h) > cr/2 - (1+1/c)r =
  (c/2 -1 - 1/c)r. 
\end{equation}
We already saw that the left-hand side of the above inequality is at
most $r$, so this case only arises when $c < \sqrt{6}+2$. 
Note also that
\[\dist_G(v,h) \le \dist_G(v,u_h)+\dist_G(u_h,h)\leq cr/2+ r=
(c/2+1)r,\]
and hence $B_{(c/2+1)r}(v)$ contains a shortest path $P_h$ from $u_h$ to $h$. 
Equivalently, $P_h$ is contained in $B_{c\tilde{r}_3}(v)$ for 
$\tilde{r}_3=(1/2+1/c)r$.
Observe that by \eqref{eq:lenPh}, the length of $P_h$ is greater than  
\[ (c/2-1-1/c)r \geq \tilde{r}_3=(1/2+1/c)r, \]
as $c \geq 4$. The length of $P_h$ is of course also bounded by
$cr/2$, the maximum length of $Q_h$, and hence~$h \in \tilde{W_3}$. 

In conclusion, we showed that $W \subseteq \tilde{W}_1 \cup
\tilde{W}_2 \cup \tilde{W}_3$, and hence $W$ contains at most $3sk$
elements by \autoref{lem:balls_aux}.
\end{proof}

We have now determined all the properties of approximate core hubs that we need 
in order to prove that any set $X_T$ has low doubling dimension. Recall that for 
this we need to show that we can cover any ball $B$ of radius $2r$ in the metric 
defined by $X_T$ by a bounded number of balls of half the
radius~$r$. We first prove a slightly weaker result in which we show
that core hubs in a ball of radius $cr/2$ can be covered by a small
number of balls of radius $2r$, for some given $r$ (note that, for~$c>4$, $2r$
is smaller than $cr/2$). We will later apply the next lemma recursively
in order to obtain a bound on the doubling dimension of $X_T$.

\begin{lem}\label{lem:cover_1level}
For any level $i$ and any ball $B_{cr_i/2}(v)\subseteq V$ of radius $cr_i/2$ we 
can cover $B_{cr_i/2}(v)\cap X_T$ with at most 
$O\left(ks\log(1/\eps)/\lambda\right)$ balls in $V$ of radius $2r_i$ each, for 
any $0<\eps\leq 2$ and violation $\lambda>0$.
\end{lem}
\begin{proof}
Recall that $X_T=\bigcup_{l=1}^{j-1} X_T^l$, where $j$ is the level of the town 
$T$ and $X_T^l$ are the approximate core hubs at level $l$ of $T$. We 
distinguish three cases based on the level $l$. 
First consider the vertices in $\bigcup_{l=1}^i X_T^l$ 
up to level~$i$, and recall that $X_T^i\subseteq \bigcup_{l=1}^i 
(C_l\cap\spc(r_l))$, i.e., the approximate core hubs of level $i$ are core hubs 
of levels up to $i$. By \autoref{dfn:core} the cores of town $T$ form a chain---
$C_{q-1}\subseteq C_q$---and thus every vertex of $\bigcup_{l=1}^i X_T^l$ 
is contained in the core $C_i$ of $T$ on level $i$. The core $C_i$ is part of 
the sprawl of level~$i$, which by \autoref{dfn:towns} is covered by balls of 
radius $2r_i$ centered at hubs in $\spc(r_i)$. For such a ball to cover some 
parts of the core $C_i$ in $B_{cr_i/2}(v)$, its center $v$ must be at 
distance at most $2r_i$ from $B_{cr_i/2}(v)$.
Hence by \autoref{lem:hub_bound} there are at most $3ks$ balls 
of radius $2r_i$ covering all of $\bigcup_{l=1}^i X_T^l$ in $B_{cr_i/2}(v)$.

Second, consider the approximate core hubs on levels $q\in\{i+1,\ldots,l\}$ 
where $l=i+\lceil\log_{c/4}(c/\eps)\rceil$. Cover every vertex of 
$\bigcup_{q=i+1}^{l} X_T^q$ in $B_{cr_i/2}(v)$ by one ball of radius $2r_i$ 
each. For any such level $q > i$ the radius of $B_{cr_i/2}(v)$ is at most 
$cr_q/2$. Since we assumed that $\eps\leq 2$ while $c>4$, the approximate core 
hubs on level $q$ are shifted by at most $\eps r_q\leq 2 r_q< cr_q/2$ to lower 
level core hubs by \autoref{alg:xt}. Hence we can bound the number of such hubs 
in $B_{cr_i/2}(v)$ per level by~$3ks$ using \autoref{lem:hub_bound}, which 
also bounds the number of balls we use to cover them. If the violation 
$\lambda$ tends to zero, the number of such levels is 
$O(\log_{c/4}(c/\eps))=O(\log(1/\eps)/\lambda)$, since 
$\log(c/4)=\log(1+\lambda/4)=\Theta(\lambda)$. In total this makes 
$O\left(ks\log(1/\eps)/\lambda\right)$ balls for levels up to $l$.

For the remaining levels $l > i+\lceil\log_{c/4}(c/\eps)\rceil$ we use the fact 
that the approximate core hubs are locally nested by \autoref{lem:nested}. In 
particular, note that $\eps r_l\geq cr_i$ since $r_l=(c/4)^l$, i.e., the 
diameter of $B_{cr_i/2}(v)$ is at most $\eps r_l$ for level $l$. Let $q$ be the 
lowest level for which $X^q_T\cap B_{cr_i/2}(v)\neq\emptyset$. If $q\leq l$ the 
hubs in $X^q_T\cap B_{cr_i/2}(v)$ are already accounted for. Otherwise, as 
before we greedily cover each hub in $X^q_T\cap B_{cr_i/2}(v)$ by a ball of 
radius $2r_i$ each, and by \autoref{lem:hub_bound} we need at most $3ks$ balls 
to do so. Now, by \autoref{lem:nested}, every vertex of $X^p_T\cap 
B_{cr_i/2}(v)$ for a level $p>\max\{l,q\}$ is contained in some set 
$X^{p'}_T\cap B_{cr_i/2}(v)$ for $p'\leq\max\{l,q\}$. Since we already covered 
each hub in $X^{p'}_T\cap B_{cr_i/2}(v)$ with a ball, the claim follows.
\end{proof}

We can now use the above lemma recursively to cover the set $X_T$ in a ball 
$B_{2r}(v)$ 
with balls of half the radius, as we show next.

\begin{lem}\label{lem:cover_levels}
Let $T\in\mc{T}$ be a town and let $B_{2r}(v)\subseteq V$ be a ball of 
radius $2r$. Then $B_{2r}(v)\cap X_T$ can be covered by at most
$\left(ks\log(1/\eps)/\lambda\right)^{O(1/\lambda)}$
balls in $V$ of radius $r$, for any $0<\eps\leq 2$ 
and violation~$\lambda>0$.
\end{lem}
\begin{proof}
Let $l$ be the smallest level for which $cr_l/2\geq 2r$. Instead of using 
$B_{2r}(v)$ directly, we will cover the larger set $B_{cr_l/2}(v)\cap X_T$ with 
balls of radius $cr_{l-1}/4 < r$, which we find by recursively covering 
$B_{cr_l/2}(v)$ with balls of the next lower level. 

Since $r_i=(c/4)^i$, a ball $B_{2r_{i+1}}(h)$ has radius 
$2r_{i+1}=2cr_i/4=cr_i/2$. Hence, by \autoref{lem:cover_1level}, we can cover 
$X_T\cap B_{2r_{i+1}}(h)$ with $O\left(ks\log(1/\eps)/\lambda\right)$ balls of 
radius $2r_i$, on which we  recurse. By the choice of $l$, $r>cr_{l-1}/4$, and 
since $r_i=(c/4)^i$, the number of levels $\beta$ on which we need to recurse is 
at most
$$
\log_{c/4}(cr_{l}/2)-\log_{c/4}(cr_{l-1}/4) =
1+\frac{1}{\log_2(c/4)}=
O(1/\lambda).
$$
The total number of balls needed to cover $B_{2r}(v)$ with balls of radius $r$ 
is then 
at most
$$
\sum_{i=0}^{\beta-1} O\left(ks\log(1/\eps)/\lambda\right)^i
= \left(ks\log(1/\eps)/\lambda\right)^{O(1/\lambda)}, %
$$
which concludes the proof.
\end{proof}

The balls $B_r(h)$ found in \autoref{lem:cover_levels} are centered at hubs. If 
all these hubs are part of $X_T$, then we have shown that $X_T$ has bounded 
doubling dimension. However, if $h \notin X_T$ for some ball center, then we 
have partly covered $B_{2r}(v)\cap X_T$ with invalid balls that are not centered 
at points in the metric~$X_T$. We already addressed this issue in 
\autoref{sec:construction} by proving \autoref{lem:set-d-dim}.
Thus we are finally ready to prove the remaining part of 
\autoref{thm:doubling_dim} by bounding the doubling dimension of $X_T$. Consider 
a ball $B_{2r}(v)\subseteq V$. According to \autoref{lem:cover_levels} we can 
cover $B_{2r}(v)\cap X_T$ using at most 
$(ks\log(1/\eps)/\lambda)^{O(1/\lambda)}$ balls in $V$ of radius~$r$. Recall 
that the doubling dimension is $\log_2 \delta$, where $\delta$ is the number of 
balls needed. Hence by \autoref{lem:set-d-dim} the doubling dimension of $X_T$ 
is $O(\log(\frac{ks\log(1/\eps)}{\lambda})/\lambda)$, as claimed.

\section{The treewidth of the embedding}
\label{sec:treewidth}

We prove by induction that the embedding has bounded treewidth. That is, we 
prove that the embedding of any town $T\in\mc{T}$ has bounded treewidth, 
assuming that the embeddings of its child towns have bounded treewidth. In 
particular, we prove the following, which implies the treewidth bound of 
\autoref{thm:main}, since there are 
$O(\log_{c/4}\alpha)=O(\log(\alpha)/\lambda)$ levels in total, and we can 
assume that $s=O(k\log k)$ by \cite{abraham2011vc}.

\begin{thm}\label{lem:treewidth}
The embedding constructed for a town $T\in\mc{T}$ of level $j$ has treewidth 
\[
j\cdot (\log(\alpha))
^{O\left(\log^2(\frac{ks}{\eps\lambda})/\lambda\right)}.
\]
\end{thm}

To prove \autoref{lem:treewidth}, we show how to compute a tree decomposition 
$D_T$ of the embedding~$H_T$, when $T$ has child towns in the towns 
decomposition. Recall that $H_T$ is obtained by connecting the embeddings 
$H_{T'}$ of each child town $T'$ to the embedding $H_X$ of the approximate core 
hubs~$X_T$. In particular, an edge is added between every vertex in $T'$ and 
every hub in the connecting bag $b$ of $T'$ in the tree decomposition $D_X$ 
of~$H_X$. To compute $D_T$ we will join the tree decompositions of the child 
towns with $D_X$. For this we need to inductively specify a root bag for each 
tree decomposition, and the root bag of $D_T$ is the highest level bag 
of~$D_X$. 

Now for each child town $T'$, consider appending the subtree $D_{T'}$ to $D_X$ 
by adding the root bag of $D_{T'}$ as a child of the connecting bag $b$ of $T'$ 
in $D_X$. This satisfies condition~\eqref{item:tw-union} of 
\autoref{dfn:treewidth}, as the union of all bags is $T$. Unfortunately, 
though, this initial tree of bags $D_T$ does not satisfy the remaining 
requirements of a valid tree decomposition of~$H_T$ according to 
\autoref{dfn:treewidth}: the edges added to connect the child towns and their 
connecting bags may not be contained in any 
bag---violating~\eqref{item:tw-edges}---and there might be some vertex $v$ for 
which the bags containing~$v$ are not connected 
in~$D_T$---violating~\eqref{item:tw-vertices}. 

To make 
$D_T$ valid we change the initial tree of bags in two steps, of which the 
first will guarantee that~\eqref{item:tw-edges} is satisfied, and the second 
guarantee that~\eqref{item:tw-vertices} is satisfied. Namely, we perform the 
following for every child town $T'$ and its connecting bag $b$ in~$D_X$:
\begin{enumerate}
\item \label{item:td_edges} add all vertices of $b$ to each bag of $D_{T'}$, and
\item \label{item:td_connected} add all hubs of $X_T\cap T'$ to each bag of $D_{T'}$, and also to $b$ and 
all descendants of $b$ in $D_X$ (but not the descendants of $b$ in $D_T$ that 
are bags of some $D_{T''}$ for some child town $T''\neq T'$ of $T$).
\end{enumerate}
We now argue that the resulting tree decomposition is valid.

\begin{lem}
\label{lem:td_edges}
After performing step \eqref{item:td_edges} above, all edges are contained within
some bag.
\end{lem}

\begin{proof}
First, note that the decompositions $D_X$ and 
$D_{T'}$ for each child town $T'$ are valid by \autoref{thm:Talwar} and by 
induction, respectively. Hence the only edges that are not contained in any 
bag of $D_T$ are those added to connect a child town $T'$ and its 
connecting bag $b$. We add all vertices of $b$ to every bag of the 
decomposition~$D_{T'}$, so after repeating this for every child town, 
for every edge in $E(H_T)$ there is a bag in $D_T$ containing both endpoints. 
\end{proof}

We will bound the growth of the bags during this step later on using the bound 
on the size of each bag $b$ of~$D_X$ given by \autoref{thm:Talwar}.
Next we show that performing the second step will guarantee 
that~\eqref{item:tw-vertices} of \autoref{dfn:treewidth} is satisfied. 

\begin{lem}
After performing step \eqref{item:td_connected} above, for all vertices $v$, the set of bags
containing $v$ form a connected subtree of $D_T$, 
and $D_T$ is a valid tree decomposition of $T$. 
\end{lem}

\begin{proof}
Suppose 
there is a vertex $v$ such that the bags containing $v$ are not connected 
after performing the first step. 
By \autoref{thm:Talwar} and by induction, the sets of bags containing each 
vertex are connected within $D_X$ and $D_{T'}$ for all child towns $T'$, so $v$ 
must be in $X_T \cap T'$ for some $T'$. This means that $v$ is an approximate 
core hub of $T$ that happens to lie in the child town~$T'$. Since child towns of 
$T$ are disjoint by \autoref{lem:laminar-towns}, $v$ cannot be contained in two 
different ones, so that $T'$ is the only child town containing $v$. Note that 
$v$ cannot be in the connecting bag $b$ of $T'$ because then the first step 
would have added $v$ to all bags of $D_{T'}$, which would have connected the 
sets of bags in $D_X$ and $D_{T'}$ containing $v$. Hence it can only be that $v$ 
is in a bag of $D_{T'}$ and in some bag of $D_X$ other than the connecting bag 
of~$T'$. 

We know from \eqref{item:net-hierarchy} in \autoref{thm:Talwar} that the 
vertices in the bags of the decomposition $D_Y$ for the representative hubs 
$Y_T$ of $T$ form a hierarchy: every vertex in a bag $b'$ of $D_Y$ is 
also contained in one of the child bags of $b'$. Recall that the 
decomposition $D_X$ of $X_T$ is obtained from $D_Y$ by simply replacing each 
vertex with all hubs it represents. Hence the vertices in the bags of $D_X$ also 
form a hierarchy. Furthermore, all hubs in $X_T\cap T'$ are in the same bags in 
$D_X$, since they are represented by the same vertex of $Y_T$. Since~$v\in 
X_T\cap T'$ is not yet in the connecting bag $b$ of~$T'$, this means that 
in $D_X$ none of the hubs in $X_T\cap T'$ are in a bag on a higher level than 
$b$.

Recall that we choose the connecting bag $b$ so that its corresponding cluster 
contains the closest approximate core hub $h$ to $T'$. In this case, $X_T\cap 
T'\neq\emptyset$ as it contains $v$, so $h$ is a hub in $X_T\cap T'$. By the 
construction of $D_X$, if $b$ contains $h$ then $b$ contains the entire set 
$X_T\cap T'$. By \eqref{item:partition} of \autoref{lem:splittree}, on each 
level the clusters for $Y_T$ partition $Y_T$. Clearly this is also true for 
$X_T$. Hence any hub of $X_T\cap T'$, including the problematic vertex $v$, can 
only be contained in bags of the decomposition~$D_X$ that are descendants of 
$b$. 

Due to these observations we add all hubs of $X_T\cap T'$ to each bag 
of~$D_{T'}$ and also to $b$ and all descendants of $b$ in $D_X$, and this will
ensure there will not be any $v$ for which the bags containing it are disconnected
in the resulting decomposition.
Note that we 
do not need to add these hubs to descendants of $b$ in $D_T$ that are bags of 
some $D_{T''}$ for some other child town $T''\neq T'$.

For the second part of the lemma, note that adding nodes to bags does not
break conditions~\eqref{item:tw-union} or~\eqref{item:tw-edges} of
\autoref{dfn:treewidth} established in \autoref{lem:td_edges},
so the resulting tree decomposition is valid.
\end{proof}

At this point we have a valid tree decomposition $D_T$, 
but we still need to bound the sizes of the resulting bags 
in $D_X$ and each $D_{T'}$. 
We use the following two lemmas to bound the size of the bags of $D_X$.
In the first we show that for each bag $b$ of $D_X$, the 
number of child towns connecting to~$b$ and containing approximate core hubs is 
bounded. In the second lemma we prove a bound on the maximum number of 
approximate core hubs in each child town.

\begin{lem}\label{lem:nr_children}
Let $b$ be a bag of the decomposition $D_X$ of the embedding $H_X$ for $X_T$, 
and let $d$ be the doubling dimension of $X_T$. The number of child towns $T'$ 
of $T$ for which $X_T\cap T'\neq\emptyset$ and for which~$b$ is their 
connecting bag, is~$O((d/\eps)^{d})$.
\end{lem}
\begin{proof}
Let $Y\subseteq Y_T$ be the set containing exactly one representative for each 
of child town $T'$ that has $b$ as its connecting bag and for which $X_T\cap 
T'\neq\emptyset$. We can bound the size of $Y$ in order to bound the 
desired number of child towns. To prove the bound we will use the fundamental 
property of low doubling dimension metrics given by \autoref{lem:dd_bound}, 
which says that such metrics have a bounded number of vertices in terms of their 
aspect ratio. We will use this lemma to bound the size of $Y$ by deriving a 
bound on its aspect ratio: since the child towns connect to the same bag $b$, we 
are able to obtain an upper bound on the distance between the representatives in 
$Y$. We also get a lower bound on the distances from the fact that $b$ was 
chosen for a child town according to the minimum distance to any other child 
town.

More concretely, consider the tree decomposition $D_Y$ for the representative 
hubs~$Y_T$. The bag $b$ was obtained from a bag $b'$ of $D_Y$ by replacing each 
vertex with the represented hubs of~$X_T$. If the level of bag $b'$ is $\bar l$ 
then, by \eqref{item:diameter} of \autoref{lem:splittree}, the diameter of the 
cluster $C'$ corresponding to $b'$ is at most~$2^{\bar l+1}$.

Suppose $T'$ is a child town that has $b$ as its connecting bag and for which 
$X_T\cap T'\neq\emptyset$. The bag $b$ was chosen so that the corresponding 
cluster contains the closest hub $h$ of~$X_T$. Since $X_T\cap T'\neq\emptyset$, 
this means $h\in X_T\cap T'$. Analogous to the connecting bag~$b$, its cluster $C$
is obtained from cluster $C'$ by replacing each vertex with its represented 
hubs. Hence all of $X_T\cap T'$ resides in $C$. 
Accordingly, the representative for the set $X_T\cap T'$ of 
each considered child town $T'$ is in $C'$, i.e.,\ $Y\subseteq C'$.

Bags $b$ and $b'$ are at the same level $\bar{l}$. Recall that we chose this 
level in the following way: if the closest sibling of a child town~$T'$ is at a 
distance in the interval $(r_i,r_{i+1}]$, then the level $\bar l$ of $b$ is 
$\min\{\bar j,\bar i+\lceil\log_{2}(d/\eps)\rceil\}$, where $\bar j$ is the 
level of the root of $D_X$ and $\bar i=\lceil\log_2 r_i\rceil$. Let $\bar 
i'=\bar l-\lceil\log_{2}(d/\eps)\rceil$ so that $\bar i'\leq \bar i$. Thus the 
distance from $T'$ to any of its siblings is more than~$r_i\geq 2^{\bar i-1}\geq 
2^{\bar i'-1}\geq\eps 2^{\bar l-1}/d$.

Since each vertex of $Y$ is in a different child town, 
the distance between any pair of vertices in $Y$ is more 
than~$\eps 2^{\bar l-1}/d$, so the aspect ratio of the set $Y$ is at most 
$2^{\bar l+1}d/(\eps 2^{\bar l-1})= O(d/\eps)$, due to the bound on the 
diameter of cluster $C'$ containing $Y$. By \autoref{lem:dd_bound} we then get 
$|Y|\leq O((d/\eps)^d)$, and this bound is the same for the number of considered 
child towns.
\end{proof}

Next we prove that the number of approximate core hubs in each child town is 
bounded. This result will also help in bounding the treewidth of $H_X$, since 
it gives a bound on the number of approximate core hubs that a vertex from 
$Y_T$ 
represents.

\begin{lem}\label{lem:nr_reps}
For any child town $T'$ of $T$, the number of approximate core hubs in the 
intersection $X_T\cap T'$ is $O\left(s\log(1/\eps)/\lambda\right)$.
\end{lem}
\begin{proof}
  Suppose that $T'$ is a town on level $i$, and 
  recall from \autoref{sec:dd} that
  \begin{equation}\label{eq:hubs}
    X_T^i\subseteq \bigcup_{q=1}^i C_q\cap\spc(r_q),
  \end{equation}
  i.e.\ the approximate core hubs of level $i$ are core hubs on levels
  $i$ or below. By \autoref{dfn:core} no such core hubs exist, and
  hence $T'$ also does not contain any approximate core hubs of level
  at most $i$. 

  Let $l = i+ \lceil \log_{c/4}(1/\eps) \rceil$, and consider
  $q \in (i,l]$.
  Once more, since $T'$ does not contain core hubs of level at most~$i$, any 
approximate core hub of level $q$ must also be
  a core hub of level $l' \in (i,q]$, and hence we focus on bounding
  the size of $\spc(r_{l'}) \cap T'$ for each $l' \in (i,l]$. Recall that
  \autoref{lem:townproperties} implies that town $T'$ 
  has diameter at most $r_i \leq cr_{l'}/2$, and therefore $T'$ is contained in
  $B_{cr_{l'}/2}(v)$ for any $v \in T'$. 
  \autoref{dfn:spc} implies that $|B_{cr_{l'}/2}(v)\cap\spc(r_{l'})|\leq s$, and
  hence also $T'$ contains no more than $s$ level $l'$ core hubs. 
  In summary, we have just shown that the set 
  \[ X= T' \cap \bigcup_{q \leq l} X^q_T \]
  has cardinality at most $s\lceil \log_{c/4}(1/\eps)\rceil$. It remains to 
consider levels $q>l$.
  Yet again by \autoref{lem:townproperties}, $T'$ has diameter at most
  \[ r_i = \left(\frac{c}{4}\right)^i \leq \eps \left( \frac{c}{4}
  \right)^{l} < \eps r_q. \]
  \autoref{lem:nested} directly implies that any approximate core hub
  in $T'$ of level greater than $l$ is contained in $X$ if the
  latter set is non-empty. So let us assume that $X=\emptyset$.  In
  this case we argue as before, and use \autoref{dfn:spc} to bound
  $|\spc(r_q) \cap T'|$ by $s$.
  All in all, we showed that $T'$ contains 
  $O(s\log_{c/4}(1/\eps))$ approximate core hubs. 
\end{proof}

Using the obtained bounds in the above lemmas, we are now ready to prove that 
the treewidth of the embedding $H_T$ is bounded.

\begin{proof}[Proof of \autoref{lem:treewidth}]
Towns that have no children are singletons, since every vertex is a town on 
level~$0$. Hence for these the claim is trivially true. Otherwise, by 
\autoref{lem:laminar-towns}, a town has at least two children. For these we need 
to bound the resulting bag sizes of the tree decomposition $D_T$, as described 
in this section. First off we determine the treewidth of the embedding $H_X$ for 
$X_T$. The decomposition $D_X$ was obtained from the decomposition $D_Y$ for 
$Y_T$ by replacing each vertex with the hubs of $X_T$ it represents. For each 
vertex of $Y_T$ the number of represented hubs is bounded by 
\autoref{lem:nr_reps}, while the treewidth of the embedding for $Y_T$ is bounded 
by \autoref{thm:Talwar}. Thus if the doubling dimension of $Y_T$ is $d$ then the 
treewidth $t_X$ of $H_X$ is 
\[
t_X\leq(d\log(\alpha)/\eps')^{O(d)}\cdot s\log(1/\eps)/\lambda.
\]

In the first step of the transformation to make the tree decomposition $D_T$ 
valid, we add all vertices of a bag $b$ of $D_X$ to all bags of the 
decomposition trees $D_{T'}$ of child towns $T'$ for which $b$ is the connecting 
bag. By \autoref{lem:laminar-towns}, if $T$ is a town on level $j$ then each of 
its child towns is on some level $i\leq j-1$. Hence if, by induction, the 
treewidth of some embedding $H_{T'}$ was $i\cdot t_X$, then it is at most 
$j\cdot t_X$ after adding the vertices of~$b$.

In the second step of the transformation of $D_T$, we add all hubs of $X_T\cap 
T'$ to every bag of~$D_{T'}$. By \autoref{lem:nr_reps}, $|X_T\cap T'|\leq 
O(s\log(1/\eps)/\lambda)$ for any child town $T'$. This term is dominated by the 
asymptotic bound on $t_X$. The second step also adds the hubs of $X_T\cap T'$ to 
the connecting bag $b$ and all descendants of $b$ in $D_X$. Note that this does 
not affect the bags of a decomposition $D_{T''}$ of any child town $T''\neq T'$ 
of $T$. By \autoref{lem:nr_children}, each bag $b$ of $D_X$ receives approximate 
core hubs from $O((d/\eps)^d)$ child towns for which $b$ is the connecting bag. 
Each such child town adds $O(s\log(1/\eps)/\lambda)$ hubs to $b$ by 
\autoref{lem:nr_reps}. Hence the total number of hubs added to $b$ from child 
towns having $b$ as their connecting bag is $O((d/\eps)^d\cdot 
s\log(1/\eps)/\lambda))$. However these hubs are also added to all descendants 
of such a bag $b$. The total number of levels of the decomposition tree $D_X$ is 
$O(\log\alpha)$ by~\eqref{item:levels} of \autoref{lem:splittree}. Hence any 
bag of $D_X$ receives at most $O((d/\eps)^d\log(\alpha)\cdot 
s\log(1/\eps)/\lambda))$ additional hubs from all its ancestors. This term is 
again dominated by the asymptotic bound on~$t_X$, since~$\eps'=\eps^2$.

It follows that the treewidth of $D_T$ is $j\cdot O(t_X)$. Hence to
conclude the proof we only need to bound $t_X$. The doubling dimension
$d$ of $Y_T\subseteq X_T$ is
$O(\log(\frac{ks\log(1/\eps)}{\lambda})/\lambda)$ by
\autoref{thm:doubling_dim}. Since
$x\cdot(\log x)^{O(\log x)}\subseteq(\log x)^{O(\log x)}$,
$(x\log x)^{O(1)}\subseteq x^{O(1)}$, and $O(\log x)\subseteq O(x)$,
the treewidth $t_X$ of $H_X$ is at most
$\log(\alpha)^{O(\log^2(\frac{ks}{\eps\lambda})/\lambda)}$.
\end{proof}

\section{Obtaining approximation schemes}
\label{sec:ptas}

In this section we demonstrate how we can use the embedding of 
\autoref{thm:main} to derive QPTASs for various  network design problems when 
the input graph $G=(V,E)$ is an edge-weighted graph with low highway dimension. 
Specifically, we consider the Travelling Salesman, Steiner Tree and Facility 
Location problems. We begin by defining these (see also~\cite{Vazirani01book}), 
and we briefly mention how these problems historically arose in contexts given 
by transportation networks. 

For the \emph{Travelling Salesman} problem the shortest \emph{tour}, i.e.\ 
cycle in the shortest-path metric, visiting all vertices of $G$ needs to be 
found. One of the earliest references\footnote{For historical references see 
\citet{schrijver2005history} and \citet{cook2012TSP}.} to the Travelling 
Salesman problem appears in a manual of 1832, in which five tours through German 
cities are suggested to a traveling salesman. The problem became known as the 
``48 States Problem of Hassler Whitney'' in 1934 after Whitney studied it in the 
context of finding the shortest route along the capitals of the lower 48 US 
states. Later milestones in its study include computing the shortest routes 
through an increasing number of cities in countries such as the USA, Germany, 
and Sweden (though these instances used Euclidean distances).

In the \emph{Steiner Tree} problem, in addition to $G$ a set of \emph{terminals} 
$R\subseteq V$ is given. The aim is to find a minimum cost tree in $G$ spanning 
all terminals (a so called Steiner tree). An early reference\footnote{For 
historical references see \citet{brazil2014history}.} to the Steiner Tree 
problem appears in a letter by Gauss from 1836, who mentioned it in the context 
of connecting cities by railways. The problem was later popularized by the book 
``What is Mathematics?'' in 1941 by Courant and Robins, who described it in 
terms of minimizing the total length of a road network. 

The \emph{Facility Location} problem assumes additional weights on the vertices, 
and the goal is to select a subset of vertices $W\subseteq V$ (the 
\emph{facilities}). The \emph{opening cost} of a facility is given by its vertex 
weight, and the \emph{connection cost} of a vertex $v\in V$ is the distance from 
$v$ to the closest facility in~$W$. The objective is to minimize the sum of all 
opening and connection costs. The Facility Location 
problem has the same root\footnote{For historical references see 
\citet{smith2009history}.} as the Steiner Tree problem in the Fermat-Torricelli 
problem from 1643, in which a point is to be found that minimizes the total 
distance to three other points in the plane. The generalization to an arbitrary 
number of other points became known as the Weber problem, after Alfred Weber 
studied it in 1909 in the context of finding a factory location so as to 
minimize the transportation costs of suppliers. Among other problems, Hakimi 
introduced Facility Location
to networks in 1964, and related it to finding locations for police stations in 
road networks.

The main result of this section is the following, of which we give a proof 
sketch below.

\begin{thm}\label{thm:qptas}
If the input graph $G$ has constant highway dimension $k$ with constant 
violation~$\lambda>0$, then for any constant $\eps\in(0,1]$ a 
$(1+\eps)$-approximation to each of the Travelling Salesman, Steiner Tree and  
Facility Location problems can be found in quasi-polynomial time.
\end{thm}

Our approach is similar to those used for Euclidean~\cite{Arora96polynomialtime} 
and low doubling dimension~\cite{talwar2004bypassing} metrics. Accordingly it 
can also be used for other problems, as in~\cite{Arora96polynomialtime}. The 
main idea is to compute a bounded treewidth graph from the input according to 
\autoref{thm:main}, and then optimally solve the computed graphs using known 
algorithms for which the running time can be bounded in terms of the treewidth. 
However, the treewidth bound of \autoref{thm:main} depends on the aspect ratio 
$\alpha$. To guarantee quasi-polynomial running times we therefore need to 
ensure that the aspect ratio of the input used in \autoref{thm:main} is not too 
large. We achieve this by computing a coarse net of polynomial aspect ratio for 
the input graph first. It is not too hard to show that only a small distortion 
of the optimum solution is incurred if the nets are fine enough, and we 
therefore obtain approximation schemes for the input instances. However, it is 
not necessarily the case that the nets themselves are shortest-path metrics of 
low highway dimension graphs, even if they are obtained from graphs of low 
highway dimension. Hence we need to argue that we can actually achieve the 
treewidth bound of \autoref{thm:main}, even though we use the nets as inputs. 

We go on to describe how a QPTAS as claimed in \autoref{thm:qptas} can be 
obtained, if a problem~$\mc{P}$ has the following properties. Thereafter we will 
show that they are true for each of our considered problems.
\begin{enumerate}
\item An optimum solution for $\mc{P}$ can be computed in time $n^{O(t)}$ for 
graphs of treewidth $t$, \label{it:tw}
\item a constant approximation to $\mc{P}$ in $G$ can be computed in 
(quasi-)polynomial time, \label{it:apx}
\item the diameter of the input graph $G$ can be assumed to be $O(n\cdot 
OPT_G)$, where $OPT_G$ is the cost of an optimum solution in~$G$, 
\label{it:diam}
\item an optimum solution in a $\delta$-net of the vertices $V$ of $G$ has 
cost at most $OPT_G+O(n\delta)$, \label{it:net}
\item the optimization function of $\mc{P}$ is linear in the edge costs, and 
\label{it:fct}
\item any solution of $\mc{P}$ in a $\delta$-net of the vertices $V$ of $G$ 
can be converted to a solution for $G$ losing at most an additive factor of 
$O(n\delta)$. \label{it:conv}
\end{enumerate}

Assuming that $\eps$, the highway dimension $k$, and the violation $\lambda$ 
are constant, the treewidth bound of \autoref{thm:main} is polylogarithmic in 
the aspect ratio $\alpha$. Combining \autoref{thm:main} with an algorithm for 
bounded treewidth graphs having a running time as proclaimed in \autoref{it:tw}, 
thus does not guarantee quasi-polynomial running time yet, since $\alpha$ might 
be large. Hence we will reduce the aspect ratio by pre-computing a coarse set of 
vertices of the input first. In particular, we greedily compute a $\delta$-net 
of $V$, where $\delta=\eps\kappa/n$ and $\kappa=\Theta(OPT_G)$ is a constant 
approximation of the cost $OPT_G$ of the optimum solution for the considered 
problem, which can be obtained according to \autoref{it:apx}. We \emph{assign} 
each vertex in $V$ to the closest point of the $\eps\kappa/n$-net. Note that 
this point is unique if we assume each shortest-path length to be unique. Since 
the minimum distance between any two vertices of the $\eps\kappa/n$-net is 
$\Omega(\eps\cdot OPT_G/n)$ and at most $O(n\cdot OPT_G)$ 
according to \autoref{it:diam}, the aspect ratio of the net is $O(n^2/\eps)$. 
For such polynomial aspect ratios, the treewidth guaranteed by 
\autoref{thm:main} yields quasi-polynomial $2^{O(\text{polylog}(n))}$ running 
times given an algorithm for bounded treewidth graphs as in \autoref{it:tw}.

Computing an embedding for the metric given by the $\eps\kappa/n$-net is not
straightforward though, since the net is not necessarily a metric given by the 
shortest-path distances of a low highway dimension graph. We will therefore use 
the structure of the input graph~$G$ and impose it on the computed net. More 
concretely, a town $T$ on level $i$ of $G$ induces a town $T'$ of level $i$ of 
the $\eps\kappa/n$-net, by restricting $T$ to the vertices of the net. All 
properties such as laminarity, separation bounds, and diameter (see 
\autoref{sec:properties}) needed for our construction are maintained by these 
subsets~$T'$. However the shortest-path covers are not maintained, since the 
hubs might not be part of the $\eps\kappa/n$-net. Instead of a shortest path 
cover, for every level $i$ we will use a set of \emph{shifted hubs}. For each 
hub in $\spc(r_i)$ of $G$ this set of shifted hubs contains the vertex of the 
$\eps\kappa/n$-net it was assigned to, which is at distance at most 
$\eps\kappa/n$.

Note that the towns decomposition of the net is given by the original hubs of 
the input graph~$G$, and not the shifted hubs. Consider the embedding that 
results from using the shifted hubs together with the imposed towns 
decomposition of the $\eps\kappa/n$-net as input to the algorithm. Apart from 
the fact that towns contain only a subset of the vertices, the only difference 
to using $G$ as input to the algorithm is that the approximate core hubs $X_T$ 
of a town $T$ are now shifted by a total of at most $\eps r_i + \eps\kappa/n$ on 
level $i$ from the original positions of the hubs in $G$. By re-examining the 
proofs of \autoref{sec:stretch} it is therefore not hard to see that in the 
embedding of the net the expected shortest-path length for any pair~$u,v$ is 
$(1+O(\eps))(\dist_G(u,v)+O(\eps\kappa/n))$, when using these hubs. By 
\autoref{it:net} the optimum solution in the $\eps\kappa/n$-net has cost at most 
$OPT_G+\eps\kappa$, and by \autoref{it:fct} the optimization function is linear 
in the edge costs. Hence by linearity of expectation, the optimum solution in 
the embedding, computed by the algorithm given by \autoref{it:tw}, has expected 
cost at most $(1+O(\eps))(OPT_G+O(\eps\kappa))= (1+O(\eps))OPT_G$. This 
solution still has to be converted into a solution of the input graph $G$, which 
can be done by \autoref{it:conv} with only an $O(\eps\kappa)$ additive 
overhead. Hence we obtain an approximation scheme.

We still need to argue that we obtain the same treewidth bound of 
\autoref{thm:main} when using shifted hubs. In particular, it might be that the 
approximate core hubs are not locally sparse, due to the additional 
$\eps\kappa/n$ shift. To argue that local sparsity can be maintained, we make 
the level $j$ for which $\eps\kappa/n\in (r_j,r_{j+1}]$ the lowest level, i.e.\ 
for any level below $j$ we remove all hubs. Note that the resulting set of hubs 
still covers all distances in the $\eps\kappa/n$-net. The total shift of a hub 
is now at most $\eps\kappa/n + \eps r_i\leq r_{j+1}+\eps r_i\leq (c/4+\eps)r_i$, 
since we made $j$ the lowest level. If we assume that $\eps\leq 1$ then this 
shift is less then $cr_i/2$. Accordingly, \autoref{lem:hub_bound} still implies 
that the hubs in $X_T$ are locally $3ks$-sparse, as needed. All other proofs 
are as before and thus we obtain the same treewidth bound as in 
\autoref{thm:main}.

Thus if all claimed properties for the considered problems are true, then this 
gives us QPTASs for low highway dimension graphs, as claimed in 
\autoref{thm:qptas}. We will go on to argue that each of the properties can be 
maintained for Travelling Salesman, Steiner Tree, and Facility Location. For 
the latter two, in addition to using a net as input instead of $G$, we also 
need to specify the additional input parameters. In particular for Steiner 
Tree, in addition to assigning each vertex of $G$ to the closest net point, we 
also need to shift terminals. More concretely, if a terminal of $R$ is assigned 
to a vertex $v$ of the net, then we make $v$ a terminal of the net. For Facility 
Location we need to adapt the opening costs in the net, which we do by setting 
the cost of a vertex $v$ in the net to the smallest cost of any vertex of $G$ 
assigned to $v$.

For each of the three problems, the linearity of the optimization function as 
required by \autoref{it:fct} is obvious from their definitions. For Travelling 
Salesman and Steiner Tree, \citet{bateni2011prize} show how to solve these 
problems in time $n^{O(1)}\cdot t^t$, where $t$ is the treewidth of the input 
instance. For Facility Location, \citet{ageev1992facility} gives an $O(n^{t+2})$ 
algorithm. This settles \autoref{it:tw}. It is well-known that a 
$2$-approximation for Travelling Salesman can be obtained from the minimum 
spanning tree~(MST), and that for Steiner Tree the MST on the metric induced by 
the terminals is a $2$-approximation~(see e.g.~\cite{Vazirani01book}). 
\citet{Mahdian06FacilityLocation} give a 1.52-approximation algorithm for the 
Facility Location problem. Hence we obtain an estimate $\kappa=\Theta(OPT_G)$ in 
each case, so that also \autoref{it:apx} is true.

It is easy to see that for any instance of the Travelling Salesman problem, 
$OPT_G$ is at least twice the diameter of the graph $G$. For Steiner Tree, 
observe that the maximum distance between two terminals is at most $OPT_G$. 
Therefore we can remove Steiner vertices (vertices that are not terminals) which 
are farther away from any terminal than $\kappa$. Thus the diameter of $G$ 
is~$O(\kappa)$. For Facility Location, consider a subgraph induced by edges of 
length at most~$\kappa$. Note that in an optimal solution, for any vertex the 
closest facility will be in its connected component in this subgraph. Hence we 
can solve the problem on each component separately. 
The diameter of such a component is at most $O(n\kappa)$. 
Therefore, we can assume that \autoref{it:diam} is true 
in each case.

The optimum Travelling Salesman tour in the net is at most $OPT_G$, since the 
net is a subset of~$V$. Since the terminals for the Steiner Tree problem are 
shifted by at most $\delta$ in a $\delta$-net, the optimum solution in the net 
has cost at most $OPT_G+n\delta$. By setting the vertex weights of the net as 
described above for the Facility Location problem, taking each facility of the 
optimum solution in $G$ and shifting it to the vertex of a $\delta$-net it is 
assigned to will increase only the total connection cost by at most $n\delta$. 
Hence the optimum solution in the net (with the adapted vertex weights) has cost 
at most $OPT_G+n\delta$. 
This shows \autoref{it:net} for each problem.

Given a solution of a $\delta$-net of a graph $G$ for Travelling Salesman, we 
obtain a tour for $G$ by making a detour from each vertex $v$ of the net to the 
vertices of $G$ assigned to $v$. The total overhead of this step is at most 
$2n\delta$. For Steiner Tree, we obtain a solution for $G$ by connecting each 
terminal in~$R$ to the terminal of the $\delta$-net it is assigned to. This 
introduces an additional cost of $n\delta$ in total. The algorithm for Facility 
Location by \citet{ageev1992facility} solves a generalization of the problem where the 
connection cost of each vertex is weighted. More concretely, in addition to the 
weight determining the opening cost, each vertex $v$ also has a weight 
$\varphi(v)$, and the connection cost of $v$ for a set $W$ of facilities is 
$\varphi(v)\cdot\dist(v,W)$. In a $\delta$-net we can set $\varphi(v)$ to be the 
number of vertices of $G$ assigned to~$v$. If a facility is opened on a vertex 
$v$ of the net, we obtain a solution to $G$ by shifting the facility to the 
vertex of smallest opening cost assigned to $v$. By our choice of the opening 
costs in the net, the total opening cost for the solution in $G$ is the same as 
for the solution in the net. Due to the additional weights~$\varphi(v)$, the 
total connection cost in the solution for $G$ is at most $n\delta$ larger than 
in the solution for the $\delta$-net. 
This shows \autoref{it:conv}, which was 
the last needed property to prove \autoref{thm:qptas}.

\ignore{
---------------------

Hence we would like to apply \autoref{thm:main} to reduce $G$ to an 
embedding with bounded treewidth. However, the embedding we construct has 
polylogarithmic treewidth in the aspect ratio $\alpha$, even if $\eps$, the 
violation~$\lambda$, and the highway dimension $k$ are constant. Thus to obtain 
quasi-polynomial running times we need to reduce the aspect ratio. In the 
following we describe how to obtain an instance from $G$ that has aspect ratio 
$O(n/\eps)$ without considerable loss in the solution quality. This implies the 
QPTAS.

As in~\cite{talwar2004bypassing}, the main idea is to compute a coarse net of 
$G$ with linear aspect ratio. In particular, we greedily compute an 
$\eps\kappa/n$-net of $V$, where $\kappa$ is an estimate of the optimum in $G$. 
We then assign each vertex in $V$ to the closest vertex of the 
$\eps\kappa/n$-net. For the Travelling Salesman problem it is 
well-known~\cite{Vazirani01book} that a $2$-approximation of the optimum tour 
length $OPT$ can be obtained from the minimum spanning tree (MST). Therefore we 
can set $\kappa$ to the cost of this approximation, so that the minimum distance 
between vertices in the $\eps\kappa/n$-net is $\Theta(\eps OPT/n)$. It is easy 
to see that for any instance of the Travelling Salesman problem, $OPT$ is at 
least twice the diameter of the graph, which implies an aspect ratio of 
$O(n/\eps)$ for the net.

Computing an embedding for the metric given by the $\eps\kappa/n$-net is not as 
straightforward. The reason is that the net is not necessarily a metric given 
by the shortest path distances of a low highway dimension graph. We will 
therefore use the structure of the input graph $G$ and impose it on the computed 
net. More concretely, a town on level $i$ of the $\eps\kappa/n$-net is a subset 
of vertices of a town of $G$ on level $i$. Clearly all properties such as 
laminarity, separation bounds, and diameter (see \autoref{sec:properties}) 
needed for our construction are maintained by these subsets. However the 
shortest path covers are not maintained, since the hubs might not be part of the 
$\eps\kappa/n$-net. Instead of a shortest path cover, for every level $i$ we 
will use a set of \emph{shifted hubs}, which for each hub in $\spc(r_i)$ of $G$ 
contains the closest vertex of the $\eps\kappa/n$-net. Hence for each pair of 
vertices $u,v$ of the $\eps\kappa/n$-net for which 
$\dist_G(u,v)\in(r_i,cr_i/2]$, there is a shifted hub $h$ on level $i$ such that 
$\dist_G(u,h)+\dist(h,v)\leq \dist_G(u,v)+2\eps\kappa/n$. 

Consider the embedding that results from using the shifted hubs together with 
the \mbox{$\eps\kappa/n$-net} as input to the algorithm. The algorithm will 
compute the approximate core hubs $X_T$ for each town $T$ given the shifted 
hubs. By closely examining the proofs of \autoref{sec:stretch} it can be shown 
that in the embedding the expected path length for any pair $u,v$ is 
$(1+O(\eps))(\dist_G(v,u)+O(\eps\kappa/n))$, when using these hubs. Clearly the 
optimum tour on the net is at most $OPT$ in $G$. Hence by linearity of 
expectation, the optimum Travelling Salesman tour on the embedding has length 
at most $(1+O(\eps))(OPT +O(\eps\kappa))= (1+O(\eps))OPT$. After computing the 
optimum tour for the $\eps\kappa/n$-net using the algorithm for bounded 
treewidth graphs, we obtain a tour for $G$ by making a detour from each vertex 
$v$ of the $\eps\kappa/n$-net to the vertices of $G$ assigned to $v$. Clearly 
the total overhead of this step is at most $2\eps\kappa$, and therefore this 
implies an approximation scheme.

We still need to argue that the above algorithm has quasi-polynomial running 
time. In particular, it might be that the shifted hubs are not locally sparse, 
which would imply that the treewidth is greater than claimed in 
\autoref{thm:main}. To argue that local sparsity can be maintained, we make the 
level $j$ for which $\eps\kappa/n\in (r_j,r_{j+1}]$ the lowest level. More 
concretely, for any level below $j$ we remove all hubs. Note that the resulting 
set of shifted hubs still covers all distances in the $\eps\kappa/n$-net. In 
addition to moving hubs by at most $\eps\kappa/n$ to obtain the shifted hubs, 
the construction of the embedding moves hubs again by at most $\eps r_i$ for 
each level $i$ to obtain the approximate core hubs $X_T$ for each town $T$. 
Hence the total shift is at most $\eps\kappa/n + \eps r_i\leq r_{j+1}+\eps 
r_i\leq (c/4+\eps)r_i$ since we made $j$ the lowest level. If we assume that 
$\eps\leq 1$ then this shift is less then $cr_i/2$. Accordingly, 
\autoref{lem:hub_bound} still implies that the hubs in $X_T$ are locally 
$3ks$-sparse, as needed to bound the treewidth of the embedding according to 
\autoref{thm:main}.

The only additional place where we rely on the local sparseness of hubs is for 
\autoref{lem:cover_1level}, in which we show that we can cover the sprawl in a 
ball of radius $cr_i/2$ with balls of radius $2r_i$. However in the proof of 
this lemma the center vertices of the smaller balls are allowed to be hubs that 
are not part of $X_T$ (due to \autoref{lem:set-d-dim}). Accordingly we can use 
the original hubs of the shortest path covers of $G$ in this argument, which are 
still locally $s$-sparse. Therefore all the bounds on the doubling dimension of 
the $X_T$ sets remain unchanged. Consequently the treewidth is polylogarithmic 
in the aspect ratio as claimed in \autoref{thm:main}, even when considering 
shifted hubs of the $\eps\kappa/n$-net. Thus this gives us a QPTAS for the 
Travelling Salesman problem on low highway dimension graphs.

\paragraph*{Steiner Tree.}

For this problem we again compute an $\eps\kappa/n$-net in order to reduce the 
aspect ratio, and assign each vertex of $G$ to the closest net point. We also 
need to shift terminals in this setting though. More concretely, if a terminal 
of $R$ is assigned to a vertex $v$ of the $\eps\kappa/n$-net, then we make $v$ 
a terminal of the net. It is easy to see that the optimum Steiner tree in the 
net has cost at most $OPT+\eps\kappa$, where $OPT$ is the cost of the optimum 
solution in $G$.

To find the minimum Steiner tree of a bounded treewidth graph, we can again use 
an algorithm of~\cite{bateni2011prize}, which also has running time 
$O(poly(n)\cdot t^t)$. After computing an optimum solution to the bounded 
treewidth embedding of the net, we obtain a solution for $G$ by connecting each 
terminal in~$R$ to the terminal in the solution of the $\eps\kappa/n$-net it is 
assigned to. This only introduces an additional $\eps\kappa$ cost in total. It 
is well-known~\cite{Vazirani01book} that the MST on the metric induced by the 
terminals is a $2$-approximation to the optimum Steiner tree. Therefore we set 
$\kappa$ to the cost of this MST, which, as above, implies an approximation 
scheme for the Steiner Tree problem.

To bound the aspect ratio of the $\eps\kappa/n$-net, observe that the maximum 
distance between two terminals is at most the cost of the MST. It is easy to see 
that w.l.o.g.\ we can remove Steiner vertices (vertices that are not terminals) 
which are farther away from any terminal than the cost of the MST. Hence the 
diameter of $G$ is bounded by $O(\kappa)$. This implies an $O(n/\eps)$ aspect 
ratio for the $\eps\kappa/n$-net. All other details to show that the resulting 
treewidth is polylogarithmic in $n/\eps$ remain the same as for the Travelling 
Salesman problem. Hence we obtain a QPTAS for the Steiner Tree problem.

\paragraph*{Facility Location.}

Again the idea is to first compute an $\eps\kappa/n$-net of the input graph 
$G$, and assign vertices of $G$ to the nearest net points. To find a 
$\kappa\in\Theta(OPT)$, where $OPT$ is the optimum solution cost of~$G$, we can 
for instance use the 1.52-approximation algorithm for the Facility Location 
problem by \citet{Mahdian06FacilityLocation}. Consider a subgraph induced by 
edges of length at most $\kappa$. Note that in an optimal solution, for any 
vertex the closest facility will be in its connected component in this subgraph. 
Hence we can solve the problem on each component separately. Since the diameter 
of such a component is at most $O(n\kappa)$, the aspect ratio can be assumed to 
be $O(n^2/\eps)$.

One issue when solving the problem on a $\eps\kappa/n$-net is that it might be 
that the vertex weights of the net are large. Hence a solution computed for the 
net can be quite bad compared to~$OPT$. We therefore need to adapt the weights 
in the net, which we do by setting the weight of a vertex $v$ of the 
$\eps\kappa/n$-net to the smallest weight of any vertex of $G$ assigned to $v$. 
In particular, taking each facility of the optimum solution in $G$ and shifting 
it to the vertex of the $\eps\kappa/n$-net it is assigned to, will increase only 
the total connection cost by at most $\eps\kappa$. Hence the optimum solution in 
the net (with the adapted vertex weights) has cost at most $OPT+\eps\kappa$.

\citet{ageev1992facility} gives an algorithm that solves the Facility Location problem 
on a graph with treewidth~$t$ in time $O(n^{t+2})$. Moreover his algorithm 
solves a generalization of the problem where the connection cost of each vertex 
is weighted. More concretely, in addition to the weight determining the opening 
cost, each vertex $v$ has a weight $\varphi(v)$, and the connection cost of $v$ 
for a set $W$ of facilities is $\varphi(v)\cdot\dist(v,W)$. For the 
$\eps\kappa/n$-net we set $\varphi(v)$ to the number of vertices of $G$ assigned 
to $v$. We then compute the embedding for the net and solve the problem on 
the resulting bounded treewidth graph, as usual. If a facility is opened on a 
vertex $v$ of the net, we obtain a solution to $G$ by shifting the facility to 
the vertex assigned to $v$ of smallest opening cost. Hence the total opening 
cost for the solution in $G$ is the same as for the solution in the net. Due to 
our choice of the additional weights~$\varphi(v)$, the total connection cost in 
the solution for $G$ is at most $\eps\kappa$ larger than in the solution for the
$\eps\kappa/n$-net. Hence, by the same arguments as for the Travelling Salesman 
problem, we obtain a QPTAS for Facility Location.

} %

\section{Comparing alternative definitions of the highway dimension}
\label{sec:alt-defs}

In this section we compare the different definitions of highway dimension, as 
given in~\cite{abraham2010highway,abraham2011vc,abraham2010highway2} and this 
paper. We also consider the hardness of computing the highway dimension. The 
original definition of~\cite{abraham2010highway} is the one we 
consider in the present work (with violation $\lambda=0$ in \autoref{dfn:hd}). 
In a follow-up paper~\cite{abraham2011vc} a more general definition was given 
(along with alternative notions such as the \emph{average} and 
\emph{cardinality-based} highway dimension, which we do not consider here). 
Later in~\cite{abraham2010highway2} another much more restrictive definition 
was given, under which graphs of constant highway dimension also 
have constant doubling dimension. Hence using this definition, the result of 
\citet{talwar2004bypassing} can be applied immediately to obtain a 
bounded-treewidth embedding with small distortion.

Note that this is not true for graphs of constant highway dimension according to 
\autoref{dfn:hd}: a star with unit edge lengths can use the center vertex as the 
single hub for any scale, since all shortest paths pass through it. Hence its 
highway dimension is $1$, but the doubling dimension of a star is~$\log_2 n$. 
In the following we will show that in fact a graph that has constant highway 
dimension according to~\cite{abraham2010highway2}, also has constant highway 
dimension according to \autoref{dfn:hd} if the violation is zero. Hence the 
original definition of~\cite{abraham2010highway} is a generalization of the one
used in~\cite{abraham2010highway2}. As far as we know, this has not been 
observed anywhere else yet. The highway dimension in~\cite{abraham2010highway2} 
is defined as follows.

\begin{dfn}[\cite{abraham2010highway2}]\label{dfn:hd1}
Given a shortest path $P=(v_1,\ldots, v_k)$ and $r>0$, an \emph{$r$-witness
path} $P'$ is a shortest path with length more than $r$, such that $P'$ can be 
obtained from $P$ by adding at most one vertex to each end. That is, either 
$P'=P$, or $P'=(v_0,v_1,\ldots, v_k)$, or $P'=(v_1,\ldots, v_k,v_{k+1})$, or 
$P'=(v_0,v_1,\ldots, v_k,v_{k+1})$. If $P$ has an $r$-witness path $P'$ it is 
said to be \emph{$r$-significant}, and $P$ is \emph{$(r,d)$-close} to a vertex 
$v$ if $\dist(P',v)\leq d$. The \emph{highway dimension} of a graph $G$ is the 
smallest integer $k$ such that for all $r>0$ and $v\in V$, there is a hitting 
set of size at most $k$ for the $r$-significant paths that are $(r,2r)$-close to 
$v$.
\end{dfn}

The following lemma from~\cite{abraham2010highway2} implies that an embedding 
for a graph of constant highway dimension according to \autoref{dfn:hd1} can 
easily be obtained by applying \autoref{thm:Talwar}.

\begin{lem}[\cite{abraham2010highway2}]\label{lem:hd_vs_dd}
A graph that has highway dimension $k$ according to \autoref{dfn:hd1} has 
doubling dimension at most $\log_2(k+1)$. %
\end{lem}

\autoref{lem:hd_vs_dd} is also useful to prove that graphs with constant 
highway dimension according to \autoref{dfn:hd1} also have constant highway 
dimension according to \autoref{dfn:hd}, as we show next.

\begin{lem}
A graph $G$ that has highway dimension $k$ according to \autoref{dfn:hd1} has 
highway dimension $O(k^2)$ according to \autoref{dfn:hd} for violation 
$\lambda=0$.
\end{lem}
\begin{proof}
Consider any ball $B$ of radius $4r$ around a vertex $v$ of $G$. We need to 
show that there is a hitting set of size $O(k^2)$ for all shortest paths of 
length more than $r$ entirely contained in $B$. Since the doubling dimension of 
$G$ is at most $\log_2(k+1)$ by \autoref{lem:hd_vs_dd}, there are at most $k+1$ 
balls of radius $2r$ that cover all vertices in $B$. In particular, any shortest 
path of length more than $r$ that is contained in $B$ also intersects some of 
the $k+1$ balls of radius $2r$. That is, each such shortest path has a vertex 
that is at distance at most $2r$ to some center vertex of one of the $k+1$ 
balls. 
Each of these balls has a hitting set of size at most $k$ for the 
$r$-significant paths that are $(r,2r)$-close to its respective center vertex. 
Since any shortest path of length more than $r$ is its own $r$-witness, the 
union of all these hitting sets intersects all the shortest paths of length more 
than $r$ in $B$. Hence there is a hub set of size $k(k+1)$ that hits all 
necessary shortest paths in $B$.
\end{proof}

\begin{wrapfigure}[17]{L}{0.4\textwidth}
\vspace{10mm}
\centering{\includegraphics[width=0.3\textwidth]{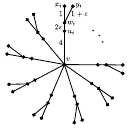}}
\caption{\label{fig:example} An example, which has highway dimension $2$ 
according to \autoref{dfn:hd2}, and for which \autoref{lem:hub_bound} is not 
true due to $B_4(v)$ and vertices $w_i$.}
\end{wrapfigure}

We now turn to the more general definition of highway dimension given 
in~\cite{abraham2011vc}. Here the idea is that the hubs need only hit shortest 
paths that pass through a ball of radius $2r$, instead of shortest paths that 
are contained in a ball of radius $4r$.

\begin{dfn}[\cite{abraham2011vc}]\label{dfn:hd2}
The \emph{highway dimension} of a graph $G$ is the smallest integer~$k$ such 
that for every scale $r>0$, and every ball $B_{2r}(v)$ of radius $2r$, there are 
at most $k$ vertices of $V$ hitting all shortest paths of length in $(r,2r]$ and 
intersecting $B_{2r}(v)$.
\end{dfn}

It is easy to see that \autoref{dfn:hd2} is a generalization of 
\autoref{dfn:hd} for violation $\lambda=0$, since any path of length at most 
$2r$ that intersects a ball $B_{2r}(v)$ is also entirely contained in the 
ball~$B_{4r}(v)$. Interestingly however, we do not know how to generalize our 
embedding results to this more general definition. In particular, we can show 
that \autoref{lem:hub_bound} does not hold for graphs of constant highway 
dimension according to \autoref{dfn:hd2}, as the next lemma implies. Hence an 
alternative method to the one developed in this paper would be needed to find 
an embedding of low distortion.

\begin{lem}
For any integer $l$ there exists a graph with highway dimension $k=2$ according 
to \autoref{dfn:hd2}, and the following properties. There is a scale $r>0$ for 
which there is a ball $B$ of radius $2r$, such that a minimal locally $2$-sparse 
shortest path cover contains $l+1$ hubs, each of which is at distance at most 
$2r$ from some vertex in $B$.
\end{lem}
\begin{proof}
Given $l$ we construct a star-like graph $G$ as follows (see 
\autoref{fig:example}). It has a center vertex $v$, and for each 
$i\in\{1,\ldots, l\}$ it has four vertices $u_i, w_i, x_i, y_i$. There is an 
edge from $v$ to $u_i$ of length~$4$, from $u_i$ to $w_i$ of length~$2\eps$, 
from $w_i$ to $x_i$ of length $1$, and from $w_i$ to $y_i$ of length $1+\eps$, 
for some suitably small~$\eps>0$.

We first prove that $G$ has highway dimension $k=2$ according to 
\autoref{dfn:hd2}. Consider a ball $B_{2r}(v)$ centered at $v$. If $r<2$ then 
this ball contains only $v$ and there is nothing to show. If $r\in [2,2+\eps)$ 
then $B_{2r}(v)=\{v,u_1,\ldots,u_l\}$, and it suffices to choose $v$ as the only 
hub for this ball: any shortest path intersecting the ball and not containing 
the hub $v$ has length at most $1+3\eps$ (e.g.~$u_1w_1y_1$), which is shorter 
than $r$. If $r\geq 2+\eps$ then $w_i\in B_{2r}(v)$ for all $i$ and the paths 
$x_iw_iy_i$ intersect the ball. It still suffices to choose $v$ as the only hub 
since a shortest path that does not contain $v$ has length at most $2+\eps$ 
(e.g.~$x_1w_1y_1$), and only paths of length more than $r$ need to be hit by the 
hubs. Now consider a ball $B_{2r}(z_i)$ for some $z_i\in\{u_i,w_i,x_i,y_i\}$. If 
$r<4$ then $B_{2r}(z_i)\subseteq\{v,u_i,w_i,x_i,y_i\}$, and it suffices to 
choose $\{v,w_i\}$ as the hub set since any path intersecting the ball passes 
through one of these vertices (if, for instance, $z_i=u_i$ and $r=2$ then this 
choice is also necessary due to $x_iw_iy_i$ and $vu_i$). If $r\geq 4$ then it 
suffices to choose only $v$ as a hub, since any shortest path not using $v$ has 
length at most $2+\eps$.

To prove that the claimed shortest path cover exists, consider the scale $r=2$, 
for which $\spc(r)=\{v,w_i\mid 1\leq i\leq l\}$. This shortest path cover is 
minimal due to the $x_iw_iy_i$ paths of length $2+\eps>r$, and the $vu_i$ paths 
of length $4=2r$, for each $i$. It is also locally $2$-sparse since the 
$B_{2r}(u_i)$ balls contain the maximum number of two hubs of $\spc(r)$. Now 
consider the ball $B:=B_{2r}(v)=\{v,u_1,\ldots,u_l\}$. Even though it contains 
only the hub $v$, each hub $w_i$ has a vertex $u_i$ in $B$ at distance 
$2\eps\leq 2r$, which proves the claim.
\end{proof}

Note that the graph constructed in the above proof does not have constant 
highway dimension according to \autoref{dfn:hd} with violation $\lambda=0$. 
This is because at scale $r=2$, the ball centered at $v$ with radius $4r$ 
contains the $x_iw_iy_i$ paths, each of which needs to be covered by a hub.

Next we observe that introducing a violation to the original definition 
of~\cite{abraham2010highway} is not an entirely innocuous change. In particular 
there are graphs for which the highway dimension grows significantly when 
changing the violation only slightly, as the following lemma shows.

\begin{lem}
\label{lem:dimexpansion}
For any constant $c>4$ there is a graph that, according to \autoref{dfn:hd}, 
has highway dimension $1$ with respect to~$c$ and highway dimension $\Omega(n)$ 
with respect to any $c' > c$.
\end{lem}
\begin{proof}
We construct a spider graph as follows.
Let $l \gg 1$ be a parameter and $G=(V,E)$ where $V=\{u,v_1,w_1, \ldots, 
v_l,w_l\}$, and $E=\{(u,v_i), (v_i,w_i) | 1\le i \le l\}$, and for all $i$ the 
lengths of $(u,v_i)$ and $(v_i,w_i)$ are $c-1$ and $1$, respectively. If $r \ge 
1$ then the hub $u$ covers all paths longer than $r$ in any ball of radius 
$cr$. Consider a ball $B_{cr}(t)$ for any vertex $t$ where $r < 1$. If $t=u$, 
the hub $u$ covers all paths in $B_{cr}(t)$ of length $(r,cr]$. If $t$ is $v_i$ 
or $w_i$ for some $i$ then $v_i$ covers all requisite paths in $B_{cr}(t)$ 
because $B_{cr}(t)$ cannot contain $v_j$ or $w_j$ for $j \neq i$. Therefore the 
highway dimension of $G$ with respect to $c$ is $1$.

On the other hand, for any $c' > c$, let $r = c/c'$ and consider the ball 
$B_{c'r}(u)$, which has radius $c'\cdot c/c' = c$ and covers the entire graph. 
Any set of hubs that covers paths of length more than $c/c' < 1$ must cover all 
edges $(v_i,w_i)$ and must therefore include $v_i$ or $w_i$ for every $i$. 
Hence the highway dimension with respect to $c'$ is at least $l = (n-1)/2$. 
\end{proof}

Finally, we also show that computing the highway dimension according to 
\autoref{dfn:hd} is NP-hard. It remains open whether this is also true when 
considering the more restrictive highway dimension definition 
from~\cite{abraham2010highway2}.

\begin{thm}
Computing the highway dimension according to \autoref{dfn:hd} is NP-hard, for 
any violation $\lambda\geq 0$, even on graphs with unit edge lengths.
\end{thm}
\begin{proof}
The reduction is from the NP-hard Vertex Cover problem~\cite{GareyJohnson}: 
given a graph $G=(V,E)$ we need to compute a minimum sized set of vertices 
$C\subseteq V$ hitting each edge, i.e.\ $v\in C$ or $u\in C$ for each $vu\in 
E$. For the reduction we introduce an additional vertex $w$ and connect it with 
every vertex in $V$. Then we give each edge of the resulting graph $G'$ unit 
length. 

A hub set hitting each shortest path of length $1$ is exactly a vertex cover for 
a graph with unit edge lengths. Note that for scale $r=1/c$, the ball 
$B_{cr}(w)$ contains all vertices of the graph $G'$. Hence removing $w$ from the 
hub set in $B_{cr}(w)$, which hits all shortest paths of length more than~$r$, 
yields a vertex cover for $G$, as $c\geq 4$. Conversely, adding $w$ to a vertex 
cover for $G$ is a hub set in $B_{cr}(w)$ hitting all necessary shortest paths. 
Thus the highway dimension according to \autoref{dfn:hd} is $k+1$ in the graph 
$G'$ if and only if the smallest vertex cover in $G$ has size $k$.
\end{proof}

\ignore{
\begin{proof}
The reduction is from 3SAT and is similar in spirit of the 
proof in~\cite{GareyJohnson} showing that the Vertex Cover problem is NP-hard. 
Intuitively, the similarity of the reductions stems from the fact that in a 
graph with unit edge lengths both a vertex cover and a hub set need to hit the 
shortest paths between adjacent vertices, namely the edges. The only difference 
is that we need to ensure that the whole reduced graph fits into one ball of 
radius $cr$ for some scale $r$.

As in~\cite{GareyJohnson}, we introduce a gadget for every variable and every 
clause of a given 3SAT formula and then connect them. For a variable $x$ we 
construct a gadget $G_x$ containing one edge $\{l,\bar l\}$ between the 
literals 
$l$ and $\bar l$ of the variable $x$. For any clause $C$ the gadget $G_C$ is a 
triangle on three vertices $v^1_C,v^2_C,v^3_C$. If a clause $C$ contains the 
literals $l_1,l_2,l_3$ we add the edges $\{v^i_C,l_i\}$ for every 
$i\in\{1,2,3\}$ between the gadgets. That is, each vertex of a clause gadget is 
adjacent to one of the literals of the clause. Additionally we introduce a 
center vertex $u$ and an edge $\{u,v^i_C\}$ for every clause $C$ and 
$i\in\{1,2,3\}$. All edges have length~$1$.

We prove that, given a 3SAT formula with $a$ clauses and $b$ literals, the 
highway dimension of the constructed graph is $2a+b+1$ if and only if the 3SAT 
formula is satisfiable. Note that a ball centered at $u$ with radius $3$ 
contains the whole graph. Hence considering the scale $r=3/c<1$, the ball 
$B_{cr}(u)$ contains all vertices, and every edge of the graph is a shortest 
path that needs to be hit by a hub in $B_{cr}(u)$, according to 
\autoref{dfn:hd}. Observe that the same hub set is also sufficient for any 
other scale and ball. Hence it is enough to show that a hub set of size 
$2a+b+1$ hitting all edges exists if and only if the 3SAT formula is 
satisfiable.

We first prove the ``only if'' direction.
Note that any hub set hitting all edges of the graph contains at least one 
vertex of every gadget~$G_x$, and at least two vertices of every gadget $G_C$. 
Hence such a hub set of size $2a+b+1$ has exactly these numbers of hubs per 
gadget, with the exception of one additional hub somewhere in the graph. 
However 
each edge adjacent to the center vertex $u$ also needs to be hit. Hence the 
additional hub needs to be $u$ (assuming w.l.o.g.\ that there are at least two 
clauses). The hubs on the variable gadgets $G_x$ define a satisfying 
assignment: 
every clause~$C$ has one vertex which is \emph{not} a hub, and this vertex 
$v^i_C$ is incident to an edge $\{v^i_C,l\}$ that needs to be hit by the hub 
$l$. Hence the literal $l$ makes the clause $C$ true.

For the other direction, consider a truth assignment for the 3SAT formula. We 
add $u$ and each literal that is part of the truth assignment to the hub set. 
For each clause gadget $G_C$ let $v^i_C$ be one of the vertices adjacent to a 
literal $l$ that is set to true, and which we know exists. We add the two 
vertices of $G_C$ different from $v^i_C$ to the hub set. The total number of 
hubs is $2a+b+1$, and every edge adjacent to~$u$ or connecting vertices within 
any gadget are clearly hit by these hubs. Any edge that connects a gadget $G_C$ 
with $G_x$ is also hit, since the only vertex of $G_C$ that is not part of the 
hub set is connected to a literal in $G_x$ that is part of the truth assignment.
\end{proof}

}
\section{Conclusions and open problems}

Our main result shows that we can find embeddings of low highway dimension 
graphs into a distribution of bounded treewidth graphs, with arbitrarily small 
expected distortion. Since the resulting treewidth is polylogarithmic in the 
aspect ratio, this implies QPTASs for several optimization problems that 
naturally arise in transportation networks. Hence, even if the network includes 
links resulting from means of transportation such as airplanes, trains, or 
buses, our results indicate that these problems are computationally easier than 
in the general case. It remains open however to determine the complexity of the 
considered problems on graphs with constant highway dimension. In particular we 
do not even know whether the problems are NP-hard for these graphs. Also, it 
remains open whether we really need the more restrictive highway dimension 
definition as given in \autoref{dfn:hd}, or whether the more general one in 
\autoref{dfn:hd2} suffices to compute an embedding.

As argued in the introduction, even a complete graph can have highway 
dimension~$1$, and therefore low highway dimension graphs do not exclude 
minors. However it is not clear whether the treewidth of such a graph can be 
bounded in terms of the aspect ratio $\alpha$. Even though the hardness results 
in~\cite{feldmann15} for the $p$-Center problem on graphs with highway dimension 
$k$ exclude treewidth bounds of the form $O(k\log\alpha)$, it is possible that 
the treewidth of such a graph is of the form $O(\log^k\alpha)$. It seems 
notoriously difficult however to either prove or disprove this.

Another interesting open problem is the possibility of finding an embedding into 
a class of graphs with a treewidth that is polylogarithmic in~$1/\eps$ but not 
the aspect ratio. This would imply PTASs for the considered optimization 
problems. One limiting factor however is that we use the embedding given by 
\citet{talwar2004bypassing} for low doubling dimension graphs in our 
construction, for which it is unclear how to obtain embeddings with treewidths 
independent of the aspect ratio. Even though \citet{bartal2012traveling} improve 
on the result by \citet{talwar2004bypassing} by giving a PTAS for the Travelling 
Salesman problem, the latter result does not give an embedding.

One alternative path to obtaining approximation algorithms is to find so called 
\emph{padded decompositions}~\cite{Abraham:2014:CRT:2591796.2591849}. Whether 
these exist for low highway dimension graphs is not known.
It may also be possible to find reductions from low highway dimension graphs 
to graphs of bounded treewidth that distort the optimal solutions of the 
instances by arbitrarily small factors. That is, the reduction would produce a 
graph on a different vertex set than the input graph, meaning that it is not an 
embedding. As for planar 
graphs~\cite{bateni2011prize,borradaile2007Steiner,klein2008TSP,ageev-facility},
this would circumvent the issue that better embeddings might not exist (as 
shown for the planar case~\cite{chakrabarti2008local,carroll2004lower}). A last 
option obviously 
would be to find algorithms that do not use algorithms for bounded treewidth 
graphs as a back-end, and instead solve the problems on the graphs directly, as 
for instance was done for Euclidean 
metrics~\cite{arora2003survey,arora1998TSP,arora1998k-median} and, in the case 
of the Travelling Salesman problem, also for low doubling 
metrics~\cite{bartal2012traveling}.

\bibliographystyle{plainnat}
\bibliography{papers}

\end{document}